\documentclass[12pt]{article}%
\usepackage{amsmath}
\usepackage{amsfonts}
\usepackage{amssymb}
\usepackage{color}
\usepackage[colorlinks,linkcolor=red,anchorcolor=blue,citecolor=blue]%
{hyperref}
\usepackage{graphicx}
\usepackage{rotating}
\usepackage{caption}
\usepackage{subcaption}%
\providecommand{\U}[1]{\protect \rule{.1in}{.1in}}
\newtheorem{theorem}{Theorem}[section]

\newtheorem{assumption}{Assumption}[section]

\newtheorem{corollary}{Corollary}[section]

\newtheorem{lemma}{Lemma}[section]

\newtheorem{proposition}{Proposition}[section]
\newtheorem{remark}{Remark}[section]

\newenvironment{proof}[1][Proof]{\noindent \textbf{#1.} }{\  \rule{0.5em}{0.5em}}
\numberwithin{equation}{section}
\begin{document}

\title{Asymptotics of Sum of Heavy-tailed Risks with Copulas}
\author{Fan Yang$^{(a)}$\thanks{\textrm{Corresponding author: fan.yang@uwaterloo.ca}%
}\qquad Yi Zhang$^{(b)}$\quad \\$^{(a)}${\small \ Department of Statistics and Actuarial Science, University
of Waterloo, }\\{\small \ Waterloo, ON N2L 3G1, Canada}\\$^{(b)}${\small \ Department of Mathematics, Zhejiang University}\\{\small Hangzhou, Zhejiang, 310027, China}}
\date{August 22, 2023}
\maketitle

\begin{abstract}
We study the tail asymptotics of the sum of two heavy-tailed random variables.
The dependence structure is modeled by copulas with the so-called tail order
property. Examples are presented to illustrate the approach. Further for each
example we apply the main results to obtain the asymptotic expansions for
Value-at-Risk of aggregate risk.

\vspace*{0.2cm} \textbf{Keywords}: tail asymptotics, heavy tail, copula,
aggregate risk

\vspace*{0.2cm} \textbf{MSC Subject Classification}: 91G70, 60G70

\vspace*{0.2cm} \textbf{Statements and Declarations: }The Authors declare that
there is no conflict of interest.

\end{abstract}


\baselineskip18pt

\section{Introduction}

Let $X$ and $Y$ be two positive heavy-tailed random variables. We are
interested in the asymptotic behavior of the tail probabilities
\begin{equation}
\Pr \left(  X+Y>t\right)  \label{tail}%
\end{equation}
for large $t${. }Asymptotic analysis of the sum of heavy-tailed random
variables is of crucial importance in applied probability and quantitative
risk management. It is the key step in the study of asymptotic expansions for
quantile based risk measures, such as Value-at-Risk (VaR), Expected Shortfall (ES) and
expectiles, which further provides a convenient way to study the behavior of
dependent extreme risks; see e.g. {Embrechts et al. \cite{ELW2009}, Embrechts
et al. \cite{ENW2009}, }Tang and Yang \cite{TY2012}, Mao and Hu
\cite{Mao-Hu-2013} and {McNeil et al. \cite{MFE2015}}. These asymptotics can
be applied in many areas of quantitative risk management such as portfolio
diversification (e.g. Alink et al. \cite{ALW-2004}, Mao and Yang
\cite{MY2015}, {Mainik and R\"{u}schendorf \cite{MR2010}, }Mainik and
Embrechts \cite{ME2013}), credit risk (e.g. Bassamboo et al. \cite{BJZ2008},
Tang et al. \cite{TTY2019}), and efficient estimation of risk measures (e.g.
Cai et al. \cite{CEDZ15}, Daouia et al. \cite{DGS18}, and Zhao et al.
\cite{ZMY21}).

{When }$X$ and $Y$ are assumed to be completely independent, both the first-
and second-order asymptotic expansions of the tail probabilities (\ref{tail})
have been well studied in the literature, for example, Barbe and McCormick
\cite{BM-2005}, Geluk et al. \cite{GDRS1997}, Geluk et al. \cite{GPD2000}, and
Mao and Hu \cite{Mao-Hu-2013}. In contrast, the study regarding dependent
risks is far less than the independent case; {the literature in this aspect
includes} the first-order analysis in a general copula setting with constraint
in Albrecher et al. \cite{Albrecher-2006}, the first-order expansions with
Archimedean copula in Alink et al. \cite{ALW-2004}, and the second-order
expansions with FGM copula in Mao and Yang \cite{MY2015}.{ In addition,
Kortschak \cite{K2012} obtained the second-order expansions for the sum of
dependent risks under a set of assumptions different from ours.}

Obviously, the dependence structure, particularly the tail dependence between
random variables $X$ and $Y$, is a crucial factor in the analysis of the tail
probabilities defined in (\ref{tail}). Copula is a convenient tool to model
dependence structure; see the monograph of Nelson \cite{Nelsen-2006} for a
complete overview of copulas. Formally, {let $C$ be a $2$-dimensional copula
and }the corresponding survival copula $\widehat{C}$ satisfies
\[
\widehat{C}(u,v)=u+v-1+C(1-u,1-v)
\]
for $0\leq u,v\leq1${. The }joint survival probability can then be represented
by
\begin{equation}
\Pr(X>x,Y>y)=\widehat{C}(\overline{F}(x),\overline{G}(y)) \label{survival}%
\end{equation}
for $x,y\geq0$. In Hua and Joe \cite{Hua-Joe-2011} the authors defined the
{tail order function }$\tau$ to capture the tail dependence,
\begin{equation}
\tau(u,v)=\lim_{t\downarrow0}\frac{\widehat{C}(ut,vt)}{t^{\kappa}\ell
(t)},\qquad u,v>0, \label{def of tail order}%
\end{equation}
{where $\ell$ is a slowly varying function}. Note that the tail order function
$\tau$ is defined up to a constant of multiplication. The parameter $\kappa$
measures {the degree of dependence. It can be shown that under the
representation of (\ref{survival}), }when $1<\kappa<2$, $X$ and $Y$ are
positively associated; {when }$\kappa=1$ and $\ell(t)\not \rightarrow 0$, they
are tail dependent; when $\kappa=2$ and $\ell(t)\geq1$, they are
\textquotedblleft near independent\textquotedblright; when $\kappa>2$, they
have negative dependent structures, which is not the focus of this work. Tail
dependence has been well studied in the literature, for example, Schmidt and
Stadtm\"{u}ller \cite{Schmidt-2006} and Kl\"{u}ppelberg et al.
\cite{Klppeberg-2007}, De Haan and Zhou \cite{DZ2011}, the series of papers of
Ledford and Tawn \cite{LT1996}, \cite{LT1997}, \cite{LT1998} and
\cite{LT2003}. By specifying the tail order function via the limit in
(\ref{def of tail order}), it has great convenience and flexibility for the
analysis of tail asymptotics. In this paper, {we are dedicated to developing
explicit second-order asymptotic expansions {of (\ref{tail}) }under this
general dependence structure.}
%

As mentioned before, an important application of the asymptotic expansions of
the tail probabilities is to derive the asymptotic expansions for quantile
based risk measures. To illustrate our results, we derive the second-order
asymptotic expansions of VaR of the aggregate risk $X+Y$ in
Section \ref{example} for each example. 
Through a detailed numerical example, we show that the asymptotic expansions
are very close to true values obtained from large size of Monte Carlo simulations.



The remainder of this paper is organized as follows. Section \ref{MR} presents
the assumptions and main results. Section \ref{example} includes examples to
illustrate the main results and applications to the VaR of aggregate risk.
Proofs of main theorems are relegated to Section \ref{proof}.

\section{Main Results\label{MR}}

In this section, we first present the major assumptions with explanations.
Then we show the main theorems in this paper.

For two positive valued functions $a(x)$ and $b(x)$, throughout this paper, we
write $a(x)\sim b(x)$ {if} $\lim_{x\rightarrow \infty}a(x)/b(x)=1$, and write
$a(x)\lesssim b(x)$ {if} $\lim \sup_{x\rightarrow \infty}a(x)/b(x)\leq1$.

\bigskip The first assumption on the heavy tailedness of the risks.

\begin{assumption}
\label{A1}Random variables $X$ and $Y$ are nonnegative and continuous with
identical distribution function $F=1-\overline{F}$ and $\overline{F}%
\in \mathrm{RV}_{-\alpha}$ for some $\alpha>0$.
\end{assumption}


By Sklar's Theorem (see e.g. Nelson \cite{Nelsen-2006} or Joe \cite{Joe-1997}%
), if the marginal distribution function $F$ is continuous, then the survival
copula $\widehat{C}$ for $X$ and $Y$ in (\ref{survival}) is unique.

The second assumption is a formal statement of the tail order function.

\begin{assumption}
\label{A2}Suppose that the survival copula $\widehat{C}(u,v)$ is symmetric,
that is, $\widehat{C}(u,v)=\widehat{C}(v,u)$ for any $0\leq u,v\leq1$. Further
assume that $\widehat{C}(u,v)$ is twice differentiable on $[0,1]^{2}$. There
exists a tail order $\kappa \in \lbrack1,2]$, a slowly varying function
$\ell(t)$, and a nondegenerate tail function $\tau(u,v)\not \equiv 0$ which is
continuously differentiable on $[0,\infty)^{2}$, such that, for any $u,v\geq
0$,%
\begin{equation}
\lim_{t\downarrow0}\frac{\widehat{C}(ut,vt)}{t^{\kappa}\ell(t)}=\tau(u,v).
\label{assumption on tail order}%
\end{equation}
Further, $\widehat{C}_{u}(u,v)=\frac{\partial}{\partial u}\widehat{C}(u,v)$
and $\widehat{C}_{v}(u,v)=\frac{\partial}{\partial v}\widehat{C}(u,v)$ are
nonincreasing in $u$ and $v$, respectively.
\end{assumption}

Note that, since $\widehat{C}(0,v)=\widehat{C}(u,0)=\widehat{C}(0,0)=0$, it is
naturally to have $\tau(0,v)=\tau(u,0)=\tau(0,0)=0$. Thus,
(\ref{assumption on tail order}) is well defined for any $u,v\geq0$. The
assumption that $\widehat{C}_{u}(u,v)$ is nonincreasing in $u$ or $\widehat
{C}_{v}(u,v)$ is nonincreasing in $v$ is the so-called stochastically
increasing property in Hua and Joe \cite{HJ2014}. Copulas with this property
have positive dependence structures.

Many widely used copulas satisfy Assumption \ref{A2}, for example the extreme
value survival copulas and {Archimedean survival copulas with regularly
varying generator. One counter example is} an {Archimedean} copula
$\widehat{C}(u,v)=uv\exp \left(  -\sigma \ln u\ln v\right)  $ with $0<\sigma
\leq1$; see the 9th copula on Page 116 of Nelson \cite{Nelsen-2006}. One
cannot find a $\kappa$ such that a nondegenerate $\tau$ exists in
(\ref{assumption on tail order}) for this copula.

In Assumption \ref{A2}, the convergence is measured for both arguments of the
survival copula approaching $(0,0)^{T}$. Such condition does not apply in the
case when one of the arguments does not approach $0$ as the other argument.
The following assumption is on the uniform convergence of the partial
derivative of $\widehat{C}$.

\begin{assumption}
\label{A3}Suppose for small $\delta>0$ there exist some $c>0$ and $t_{0}>0$,
such that for all $0<t<t_{0}$, $0\leq u\leq \delta$, and $0\leq v\leq1$,
\[
\frac{\widehat{C}_{v}(t(1+u),v)}{\widehat{C}_{v}(t,v)}-1\lesssim cu.
\]
{\large \textbf{ }}
\end{assumption}


Assumption \ref{A3} is related by the regular variation of $\widehat{C}%
_{v}(ut,v)$, which is formally defined in the next assumption.

\begin{assumption}
\label{A4}Assume that there exist a slowly varying function $h(t)$, and a
function $\varphi(u,v)$ on $(0,\infty)\times(0,1]$ such that for any $u>0$ and
$0<v\leq1$,%
\begin{equation}
\lim_{t\downarrow0}\frac{\widehat{C}_{v}(ut,v)}{t^{\theta}h(t)}=\varphi(u,v).
\label{part1}%
\end{equation}
Further assume that $\varphi(1,\frac{1}{\cdot})\in \mathrm{RV}_{\beta}$ for
some $\beta \geq0$.
\end{assumption}

The parameter $\theta$ is at least $1$, and actually in many examples
$\theta=1$. Justification for such assertions can be found in Lemma
\ref{Taylor} and Remark \ref{remark} in the sequel. \bigskip

Now we are ready{ to present the main results of the paper with proofs
relegated to Section \ref{proof}.}

\begin{theorem}
\label{TH1-P-D1}Suppose that random variables $X$ and $Y$ with survival copula
$\widehat{C}\left(  u,v\right)  $ satisfy Assumptions \ref{A1} -- \ref{A3}.
Further, assume that%
\begin{equation}
\lim_{\delta \downarrow0}\underset{t\rightarrow \infty}{\lim \sup}\frac{\int
_{0}^{\delta}y\widehat{C}_{v}\left(  \overline{F}(t),\overline{F}\left(
ty\right)  \right)  dF(ty)}{\left(  \overline{F}\left(  t\right)  \right)
^{\kappa}\ell \left(  \overline{F}\left(  t\right)  \right)  }=0. \label{A5}%
\end{equation}
If $D(\delta,t)\not =0$ for any $0<\delta<1/2$, $\eta<\infty$ and
$\alpha(1-\beta)<1$, then we have as $t\rightarrow \infty$,%
\[
\Pr \left(  X+Y>t\right)  =2\overline{F}(t)+\Delta_{1}\left(  \overline
{F}(t)\right)  ^{\kappa}\ell \left(  \overline{F}(t)\right)  (1+o(1)),
\]
where $\Delta_{1}=\left(  2\eta+\tau \left(  2^{\alpha},2^{\alpha}\right)
-2\tau \left(  1,2^{\alpha}\right)  \right)  $.
\end{theorem}

\begin{theorem}
\label{Thm2}Suppose that random variables $X$ and $Y$ with survival copula
$\widehat{C}\left(  u,v\right)  $ satisfy Assumptions \ref{A1} -- \ref{A4}.
Assume that $\frac{\widehat{C}_{v}\left(  t,v\right)  }{t^{\theta}h\left(
t\right)  }\leq g(v)$ for all $t$ close to $0$ and $0\leq v\leq1$, where $h$
is the slowly varying function in (\ref{part1}) and $g:[0,1]\rightarrow
\lbrack0,\infty]$ such that $\int_{0}^{1}g\left(  x\right)  dx<\infty$.
Further assume that
\begin{equation}
\underset{t\rightarrow \infty}{\lim \sup}\frac{\left(  \overline{F}\left(
t\right)  \right)  ^{\kappa-\theta}\ell \left(  \overline{F}\left(  t\right)
\right)  }{\int_{0}^{1/2}y\varphi \left(  1,\overline{F}\left(  ty\right)
\right)  dF(ty)h\left(  \overline{F}\left(  t\right)  \right)  }=0.
\label{cond}%
\end{equation}
If $D(\delta,t)=0$ for some $0<\delta<1/2$, or $\eta=\infty$, or $0\leq
\beta \leq1$ and $\alpha(1-\beta)\geq1$, then we have as $t\rightarrow \infty$,%
\[
\Pr \left(  X+Y>t\right)  =2\overline{F}(t)+\Delta_{2}\left(  \overline
{F}(t)\right)  ^{\kappa}\ell \left(  \overline{F}(t)\right)  (1+o(1))+2\Delta
(t)\left(  \overline{F}(t)\right)  ^{\theta}h\left(  \overline{F}(t)\right)
(1+o(1)),
\]
where $\Delta_{2}=\left(  \tau \left(  2^{\alpha},2^{\alpha}\right)
-2\tau \left(  1,2^{\alpha}\right)  \right)  $ and $\Delta(t)=\int_{0}%
^{1/2}\left(  (1-y)^{-\alpha \theta}-1\right)  \varphi(1,\overline{F}\left(
ty\right)  )dF(ty)$.
\end{theorem}

If $X$ and $Y$ are tail independent, that is $\kappa>1$ or $\kappa=1$ and
$\ell(t)\rightarrow0$, then both Theorems \ref{TH1-P-D1} and \ref{Thm2}
present the second-order asymptotics for $\Pr(X+Y>t)$. If $X$ and $Y$ are tail
dependent (i.e., $\kappa=1$ and $\ell(t)\not \rightarrow 0$), then the
second-order asymptotic term is reduced to be a part of the first-order
asymptotic term and then is combined with $2\overline{F}(t)$. Note that if
$(X,Y)$ is multivariate regularly varying (MRV), then the first-order
expansion of (\ref{tail}) is well known in the literature; see e.g. Barbe et
al. \cite{BFG2006}. For more details on MRV, the readers are referred to
{Resnick \cite{Resnick-2007}.}

{Most widely used copulas are symmetric, for example the independence copula,
Gauss copula, $t$ copula, extreme value copulas and {Archimedean copulas.
Symmetry} indeed provides mathematical convenience. There is also a
need for asymmetric copulas, for example in hydrological studies
(\cite{GS2006}) and in ocean engineering (\cite{ZKBDS2018}). The test of
symmetry for copulas can be found in, e.g. \cite{GNQ2012}. Although both
Theorems \ref{TH1-P-D1} and \ref{Thm2} are proved under the assumption that
the survival copula is symmetric, it is possible to extend the results to
asymmetric copulas by applying the same techniques with more assumptions on
both dimensions. 

The methodology used in this paper mainly works for two dimensional case. Although it is not straightforward to extend to the case of more than two random variables, it sheds  lights on how we can approach the higher dimensional case by looking into the partial derivatives of the copula and imposing the  regularity conditions.

The second-order asymptotic expansions  derived in Theorems \ref{TH1-P-D1} and \ref{Thm2} provide a fast and reliable way to approximate the tail probabilities whereas it requires a large sample size of naive Monte Carlo simulations to reach the same level of accuracy and the sample size increases dramatically for estimating low tail probabilities (i.e. large $t$). Moreover these tail probability expansions can be applied to derive the asymptotic expansions of quantile based risk measures such as VaR and ES; see Theorems \ref{varind} and \ref{EVCvar} below. Further, these risk measures' expansions are useful for constructing asymptotic normality when estimating risk measures at high confidence levels.      }

\section{Applications\label{example}}

In this section, examples are presented to illustrate our main results for two
dependence structures, independence and extreme value copulas. Further we
derive the second-order asymptotic expansions of VaR for aggregate risks. Let
$Z=X+Y$ denote the aggregate risk with distribution function $G$.  VaR of $Z$ at confidence level $q\in(0,1)$ is defined as
\[
\mathrm{VaR}_{q}(Z)=\inf \{x\in \mathbb{R}:G(x)\geq q\}.
\] 
To derive
the second-order asymptotics of VaR, we also need to assume that the marginal
distribution $F$ is second-order regular variation. To make the result more
transparent, we simply assume that there exists $\alpha>0$, $\rho<0$, $A>0$
and $B\in \mathbb{R}$ such that
\begin{equation}
\overline{F}(t)=Ax^{-\alpha}(1+Bx^{\rho}(1+o(1))). \label{2RV}%
\end{equation}

\subsection{Independence survival copula}

The study of the second-order asymptotic expansions of the tail probabilities
(\ref{tail}) under the independence copula $\widehat{C}(u,v)=uv $ is usually
carried out when the random variables $X$ and $Y$ are regularly varying and
additionally they are assumed to be asymptotically smooth or second-order
regularly varying; see, for example, Barbe and McCormick \cite{BM-2005} and
Mao and Hu \cite{Mao-Hu-2013}. The following proposition recovers these
results by using the method proposed in the current work without any
additional assumption mentioned above.

\begin{proposition}
\label{indp}Assume that the {random variables $X$ and $Y$ satisfy Assumption
\ref{A1} and have an independence} survival copula $\widehat{C}(u,v)=uv$. {As
}$t\rightarrow \infty$, we have

(i) if $\alpha<1$, then%
\begin{equation}
\Pr(X+Y>t)=2\overline{F}(t)+(2I(\alpha,\alpha)+2^{2\alpha}-2^{\alpha
+1})\left(  \overline{F}(t)\right)  ^{2}(1+o(1)), \label{ind2}%
\end{equation}
where $I(\alpha,\alpha)$ is given by%
\begin{equation}
I(\alpha,\beta)=\beta \int_{0}^{1/2}\left(  (1-y)^{-\alpha}-1\right)
y^{-\beta-1}dy; \label{Ia}%
\end{equation}

(ii) if $\alpha \geq1$, then%
\begin{equation}
\Pr(X+Y>t)=2\overline{F}(t)+2\alpha t^{-1}\mu_{F}(t)\overline{F}(t)(1+o(1)),
\label{ind1}%
\end{equation}
where $\mu_{F}(t)=\int_{0}^{t}xdF(x)$.
\end{proposition}

\begin{proof}
It is easy to see that $\widehat{C}$ satisfies Assumption \ref{A2} with
$\kappa=2$, $\tau(u,v)=uv$ and $\ell(t)=1$. Since $\widehat{C}_{v}(u,v)=u$, it
is straightforward to verify that $\widehat{C}$ satisfies Assumption \ref{A3},
and satisfies Assumption \ref{A4} with $\theta=1$, $\varphi(u,v)=u$, $h(t)=1$,
and $\beta=0$. It remains to check relations (\ref{A5}) and (\ref{cond}) and
two key terms:
\[
D(\delta,t)=\int_{\delta}^{1/2}\left(  (1-y)^{-\alpha}-1\right)  dF(ty),
\]
where $0<\delta<1/2$, and%
\[
\eta=\alpha \int_{0}^{1/2}\left(  (1-y)^{-\alpha}-1\right)  y^{-\alpha-1}dy.
\]

(i) Since $\beta=0$, we first consider the case that $0<\alpha<1$. Apparently,
$D(\delta,t)\neq0$ and $\eta<\infty$. From the proof of Lemma 5.6 in Barbe and
McCormick \cite{BM-2005}, we obtain
\[
\lim_{\delta \rightarrow0}\lim_{t\rightarrow \infty}\frac{\int_{0}^{\delta
}ydF(ty)}{\overline{F}(t)}=0,
\]
which verifies (\ref{A5}). Accordingly, relation (\ref{ind2}) follows from
Theorem \ref{TH1-P-D1}.

(ii) Next, we consider $\alpha \geq1$. In this case $\frac{\widehat{C}%
_{v}\left(  t,v\right)  }{t}=1$, which yields $\eta=\infty$. From the proof of
Lemma 5.2 in Barbe and McCormick \cite{BM-2005}, we have
\[
\overline{F}\left(  t\right)  =o\left(  t^{-1}\int_{0}^{t/2}ydF(y)\right)  ,
\]
which verifies (\ref{cond}). Note that by Lemma 2.4 of Mao and Hu
\cite{Mao-Hu-2013}, we have
\[
\lim_{t\rightarrow \infty}\frac{\int_{0}^{1/2}\left(  (1-y)^{-\alpha}-1\right)
dF(ty)}{t^{-1}\mu_{F}(t)}=\alpha.
\]
Then by Theorem \ref{Thm2}, relation (\ref{ind1}) follows.\bigskip
\end{proof}

As the application, we next show that the asymptotic expansions of the risk
measure VaR of the aggregate risk $Z=X+Y$ under the independence copula.

\begin{theorem}
\label{varind}Assume that the {random variables $X$ and $Y$ satisfy Assumption
\ref{A1} and have an independence} survival copula. Further assume that
$\overline{F}(t)$ satisfies (\ref{2RV}). Then we have as $q\rightarrow1$,

(i) if $\alpha<1$, and $\rho<-\alpha$, then%
\[
\mathrm{VaR}_{q}(Z)=2^{1/\alpha}\mathrm{VaR}_{q}(X)\left(  1+\frac
{I(\alpha,\alpha)+2^{2\alpha-1}-2^{\alpha}}{2\alpha}\left(  1-q\right)
(1+o(1))\right)  ;
\]

(ii) if $\alpha \geq1$ and $\rho<-1$, then
\[
\mathrm{VaR}_{q}(Z)=\mathrm{VaR}_{q}(X)\left(  2^{1/\alpha}+\mu_{F}\left(
F^{\leftarrow}(q)\right)  \left(  F^{\leftarrow}(q)\right)  ^{-1}%
(1+o(1))\right)  ;
\]

(iii) if $\alpha<1$ and $-\alpha<\rho<0$, or if $\alpha \geq1$ and $-1<\rho<0$,
then%
\[
\mathrm{VaR}_{q}(Z)=2^{1/\alpha}\mathrm{VaR}_{q}(X)\left(  1+B\alpha
^{-1}\left(  2^{-\rho/\alpha}-1\right)  \left(  \mathrm{VaR}_{q}(X)\right)
^{\rho}(1+o(1))\right)  .
\]

\end{theorem}

\begin{proof}
(i) From (\ref{ind2}), for any $x>0$,%
\begin{align*}
\frac{\overline{G}(tx)}{\overline{F}(t)}  &  =x^{-\alpha}\left(  1+B(x^{\rho
}-1)t^{\rho}(1+o(1))\right)  \left(  2+(2I(\alpha,\alpha)+2^{2\alpha
}-2^{\alpha+1})\overline{F}(tx)(1+o(1))\right) \\
&  =x^{-\alpha}\left(  2+\left(  2B(x^{\rho}-1)t^{\rho}+(2I(\alpha
,\alpha)+2^{2\alpha}-2^{\alpha+1})x^{-\alpha}\overline{F}(t)\right)
(1+o(1))\right)
\end{align*}
If $t^{\rho}=o\left(  \overline{F}(t)\right)  $, i.e. $\rho<-\alpha$, then
\[
\lim_{t\rightarrow \infty}\frac{\frac{\overline{G}(tx)}{\overline{F}%
(t)}-2x^{-\alpha}}{\overline{F}(t)}=(2I(\alpha,\alpha)+2^{2\alpha}%
-2^{\alpha+1})x^{-2\alpha}.
\]
By Vervaat's Lemma,%
\[
\lim_{t\rightarrow \infty}\frac{\frac{G^{\leftarrow}(1-x\overline{F}(t))}%
{t}-\left(  \frac{x}{2}\right)  ^{-1/\alpha}}{\overline{F}(t)}=\frac
{I(\alpha,\alpha)+2^{2\alpha-1}-2^{\alpha}}{\alpha}\left(  \frac{x}{2}\right)
^{1-1/\alpha}.
\]
By letting $\overline{F}(t)=1-q$ and thus $t=F^{\leftarrow}(q)$, we have%
\[
\lim_{q\rightarrow1}\frac{\frac{G^{\leftarrow}(1-x(1-q))}{F^{\leftarrow}%
(q)}-\left(  \frac{x}{2}\right)  ^{-1/\alpha}}{1-q}=\frac{I(\alpha
,\alpha)+2^{2\alpha-1}-2^{\alpha}}{\alpha}\left(  \frac{x}{2}\right)
^{1-1/\alpha}.
\]
We obtain the second-order asymptotic expansions of $\mathrm{VaR}_{q}(Z)$ by
letting $x=1$.

(ii) From (\ref{ind1}), for any $x>0$,%
\begin{align*}
\frac{\overline{G}(tx)}{\overline{F}(t)}  &  =x^{-\alpha}\left(  1+B(x^{\rho
}-1)t^{\rho}(1+o(1))\right)  \left(  2+2\alpha \left(  tx\right)  ^{-1}\mu
_{F}(tx)(1+o(1))\right) \\
&  =x^{-\alpha}\left(  2+\left(  2B(x^{\rho}-1)t^{\rho}+2\alpha x^{-1}\mu
_{F}(t)t^{-1}\right)  (1+o(1))\right)  .
\end{align*}
If $t^{\rho}=o\left(  t^{-1}\right)  $, i.e. $\rho<-1$, then the desired
result follows similarly as in (i).

(iii) If $\alpha<1$ and $\overline{F}(t)=o\left(  t^{\rho}\right)  $, i.e.
$-\alpha<\rho<0$, or if $\alpha \geq1$ and $t^{-1}=o\left(  t^{\rho}\right)  $,
i.e. $-1<\rho<0$, then%
\[
\frac{\overline{G}(tx)}{\overline{F}(t)}=x^{-\alpha}\left(  2+2B(x^{\rho
}-1)t^{\rho}(1+o(1))\right)  .
\]
The rest follows similarly as in (i).
\end{proof}

\subsection{Extreme value survival copula}

{A copula }$C_{EV}${\ is an extreme value copula, if there exists a function
}$A:[0,\infty)^{2}\rightarrow \lbrack0,\infty)$ such that
\[
C_{EV}{(u,v)=\exp}\left(  -A(-\log u,-\log v\right)  ).
\]
The function $A$ is convex, homogeneous of order $1$ and satisfies
\[
\max(u,v)\leq A(u,v)\leq u+v.
\]
Refer to Joe \cite{Joe-1997} for details of extreme value copulas.

Assume that the survival copula $\widehat{C}$ is an extreme value copula. The
joint survival probabilities are given by%
\begin{align*}
\Pr \left(  X>x,Y>y\right)   &  =\widehat{C}_{EV}{(}\overline{{F}%
}{(x),\overline{{F}}{(y)})}\\
&  ={\exp}\left(  -A(-\log \overline{{F}}{(x)},-\log{\overline{{F}}{(y)}%
}\right)  ){.}%
\end{align*}
For the convenience of presentation, we denote $A_{1}(u,v)=\frac{\partial
}{\partial u}A(u,v)$ and $A_{2}(u,v)=\frac{\partial}{\partial v}A(u,v)$. If
$A$ is symmetric, then $A_{1}(u,v)=A_{2}(v,u)$.

The following lemma is to ensure that Assumptions \ref{A2} -- \ref{A4} are
satisfied by extreme value survival copulas.

\begin{lemma}
\label{EVClm}Suppose that the survival copula $\widehat{C}$ is an extreme
value copula with a symmetric and differentiable function $A$. Further assume
there exists $\delta>0$, $0<t_{0}<1$ and $c>0$, such that for all $x\geq0$ and
$0<t\leq t_{0}$, and $0\leq u<\delta$,
\begin{equation}
\frac{A_{2}\left(  1,\frac{x}{1+\frac{\log(1+u)}{\log t}}\right)  }%
{A_{2}(1,x)}\lesssim(1+u)^{c}. \label{evcond}%
\end{equation}
Then Assumptions \ref{A2} to \ref{A4} are satisfied.
\end{lemma}

\begin{proof}
We first verify Assumption \ref{A2}. Note that%
\begin{align*}
\frac{{\exp}\left(  -A(-\log ut,-\log vt)\right)  }{\exp(A(1,1)\log t)}  &
=\frac{\exp \left(  \log t\cdot A\left(  \frac{\log u}{\log t}+1,\frac{\log
v}{\log t}+1\right)  \right)  }{\exp(A(1,1)\log t)}\\
&  =\exp \left(  \log t\left(  A\left(  \frac{\log u}{\log t}+1,\frac{\log
v}{\log t}+1\right)  -A(1,1)\right)  \right)  .
\end{align*}
Thus, for any $u,v\geq0$, we have%
\[
\lim_{t\downarrow0}\frac{\widehat{C}{(ut,vt)}}{t^{A(1,1)}}=u^{A_{1}%
(1,1)}v^{A_{1}(1,1)}.
\]

Next, we verify Assumption \ref{A3}. Note that there exits $t_{0}>0$ such that
for all $0<t<t_{0}$ and $0\leq u\leq \delta$ and $0\leq v\leq1$,
\begin{align*}
\frac{\widehat{C}_{v}{(t(1+u),v)}}{\widehat{C}_{v}{(t,v)}}  &  =\exp \left(
\log t\left(  A\left(  \frac{\log{(1+u)}}{\log t}+1,\frac{\log v}{\log
t}\right)  -A\left(  1,\frac{\log v}{\log t}\right)  \right)  \right)
\frac{A_{2}\left(  \frac{\log{(1+u)}}{\log t}+1,\frac{\log v}{\log t}\right)
}{A_{2}\left(  1,\frac{\log v}{\log t}\right)  }\\
&  =\exp \left(  A_{1}\left(  \xi,\frac{\log v}{\log t}\right)  \log
{(1+u)}\right)  \frac{A_{2}\left(  \frac{\log{(1+u)}}{\log t}+1,\frac{\log
v}{\log t}\right)  }{A_{2}\left(  1,\frac{\log v}{\log t}\right)  }\\
&  ={(1+u)}^{{A_{1}\left(  \xi,\frac{\log v}{\log t}\right)  }}\frac
{A_{2}\left(  \frac{\log{(1+u)}}{\log t}+1,\frac{\log v}{\log t}\right)
}{A_{2}\left(  1,\frac{\log v}{\log t}\right)  },
\end{align*}
where $\xi \in \left(  1+\frac{\log{(1+u)}}{\log t},1\right)  $ with its
existence ensured by the mean value theorem. By $A_{1}\left(  y,1\right)  $ is
increasing in $y$ and $0<A_{1}\left(  y,1\right)  \leq1$ for all $y\geq0$, we
have
\[
A_{1}\left(  \xi,\frac{\log v}{\log t}\right)  \leq A_{1}\left(  1,\frac{\log
v}{\log t}\right)  =A_{1}\left(  \frac{\log t}{\log v},1\right)  \leq1.
\]
Thus, by Taylor's expansion, for all $0<t<t_{0}$ and $0\leq u\leq \delta$ and
$0\leq v\leq1$,
\begin{equation}
{(1+u)}^{{A_{1}\left(  \xi,\frac{\log v}{\log t}\right)  }}\sim \left(
{1+A_{1}\left(  \xi,\frac{\log v}{\log t}\right)  u}\right)  \leq1+u.
\label{ev1}%
\end{equation}
By $A_{2}(\cdot,\cdot)$ is homogenous of order $0$ and (\ref{evcond}) it
follows that
\begin{equation}
\frac{A_{2}\left(  \frac{\log{(1+u)}}{\log t}+1,\frac{\log v}{\log t}\right)
}{A_{2}\left(  1,\frac{\log v}{\log t}\right)  }=\frac{A_{2}\left(
1,\frac{\frac{\log v}{\log t}}{\frac{\log{(1+u)}}{\log t}+1}\right)  }%
{A_{2}\left(  1,\frac{\log v}{\log t}\right)  }\leq(1+u)^{c}. \label{ev2}%
\end{equation}
Thus, combining (\ref{ev1}) and (\ref{ev2}) we obtain
\[
\frac{\widehat{C}_{v}{(t(1+u),v)}}{\widehat{C}_{v}{(t,v)}}-1\lesssim
(1+u)(1+u)^{c}-1\lesssim(c+1)u,
\]
which shows Assumption \ref{A3} holds.

Lastly, by Euler's homogeneous function theorem, we can easily obtain that
$A_{1}(1,0)=1$. According to Hua and Joe \cite{HJ2014}, $0\leq A_{2}%
(1,0)\leq1$. If $A_{2}(1,0)\not =0$, then
\begin{align}
\lim_{t\downarrow0}\frac{\widehat{C}_{v}{(ut,v)}}{t}  &  =\lim_{t\downarrow
0}\exp \left(  \log t\left(  A\left(  \frac{\log u}{\log t}+1,\frac{\log
v}{\log t}\right)  -1\right)  \right)  A_{2}\left(  \frac{\log u}{\log
t}+1,\frac{\log v}{\log t}\right)  \frac{1}{v}\nonumber \\
&  =A_{2}(1,0)uv^{A_{2}(1,0)-1}=\varphi(u,v).\nonumber
\end{align}
which implies $\theta=1$ and $h(t)=1$ in (\ref{part1}). It is straightforward
to see that $\varphi(1,\frac{1}{\cdot})\in \mathrm{RV}_{\beta}$ with
$\beta=1-A_{2}(1,0)$, which verifies Assumption \ref{A4}. If $A_{2}(1,0)=0$,
then $\varphi(u,v)=0$, which a degenerate case and Assumption \ref{A4} is
still satisfied. \bigskip
\end{proof}

\begin{remark}
The condition (\ref{evcond}) in Lemma \ref{EVClm} is satisfied by extreme
value copulas such as the comonotonicity copula with $A(u,v)=\max(u,v)$,
independence copula with $A(u,v)=u+v$ and Gumbel copula with $A(u,v)=\left(
u^{\phi}+v^{\phi}\right)  ^{1/\phi}$ with $\phi \geq1$ (the verification is
provided at the end of this subsection).\bigskip
\end{remark}

To make the statement of the next theorem clear, we use the following notation
to denote the sets of the tail index $\alpha$ and the function $A:$%
\[
C_{1}=\{(\alpha,A):\alpha<A_{2}(1,0)^{-1}\},
\]%
\[
C_{2}=\{(\alpha,A):\alpha<A_{1}(1,1)^{-1}\},
\]
and%
\[
C_{3}=\{(\alpha,A):A(1,1)<A_{2}(1,0)+1\}.
\]
The three cases in the next theorem correspond to three disjoint sets of
$(\alpha,A)$.

\begin{theorem}
\label{EVC}Under the conditions of Lemma \ref{EVClm} and Assumption \ref{A1},{
we have as }$t\rightarrow \infty$,

(i) if $(\alpha,A)\in C_{1}\cap C_{2}\cap C_{3}$, then
\[
\Pr \left(  X+Y>t\right)  =2\overline{F}(t)+\zeta_{1}\left(  \overline
{F}(t)\right)  ^{A(1,1)}(1+o(1)),
\]
where $\zeta_{1}=2I\left(  \alpha A_{1}(1,1),\alpha A_{1}(1,1)\right)
+2^{2\alpha A_{1}(1,1)}-2^{\alpha A_{1}(1,1)+1}$ and $I(\cdot,\cdot)$ is given
in (\ref{Ia});

(ii) if $(\alpha,A)\in C_{1}\backslash(C_{2}\cap C_{3})$, then%
\begin{align*}
\Pr \left(  X+Y>t\right)   &  =2\overline{F}(t)+\zeta_{2}\left(  \overline
{F}(t)\right)  ^{A_{2}(1,0)+1}(1+o(1))\\
&  +\left(  2^{2\alpha A_{1}(1,1)}-2^{\alpha A_{1}(1,1)+1}\right)  \left(
\overline{F}(t)\right)  ^{A(1,1)}1_{\{A(1,1)=A_{2}(1,0)+1\}}(1+o(1)),
\end{align*}
where $\zeta_{2}=2I\left(  \alpha,\alpha A_{2}(1,0)\right)  $;

(iii) if $(\alpha,A)\in C_{1}^{c}$, then%
\begin{align*}
\Pr \left(  X+Y>t\right)   &  =2\overline{F}(t)+2\alpha t^{-1}\mu
_{\widetilde{F}}(t)\overline{F}(t)(1+o(1))\\
&  +\left(  2^{2\alpha A_{1}(1,1)}-2^{\alpha A_{1}(1,1)+1}\right)  \left(
\overline{F}(t)\right)  ^{A(1,1)}1_{\{A(1,1)=A_{2}(1,0)+1\}}(1+o(1)),
\end{align*}
where $\widetilde{F}(t)=1-\left(  \overline{F}(t)\right)  ^{A_{2}(1,0)}$ and
$\mu_{\widetilde{F}}(t)=\int_{0}^{t}xd\widetilde{F}(x)$.
\end{theorem}

\begin{proof}
We show that extreme value copulas satisfy the conditions in Theorems
\ref{TH1-P-D1} and \ref{Thm2} by analyzing a few key terms.

For $\eta$, if $\alpha A_{1}(1,1)<1$, then%
\[
\eta=A_{2}\int_{0}^{1/2}\left(  \left(  1-y\right)  ^{-\alpha A_{1}%
(1,1)}-1\right)  y^{-\alpha A_{1}(1,1)-1}dy=\alpha^{-1}I\left(  \alpha
A_{1}(1,1),\alpha A_{1}(1,1)\right)  <\infty.
\]
If $\alpha A_{1}(1,1)\geq1$, then $\eta=\infty$. For $D(\delta,t)$,%
\[
D(\delta,t)=A_{1}(1,1)\int_{\delta}^{1/2}\left(  \left(  1-y\right)  ^{-\alpha
A_{1}(1,1)}-1\right)  y^{-\alpha(A_{1}(1,1)-1)}dF(ty)>0.
\]
For $0<A_{2}(1,0)\leq1$, we show that there exists a function $g$ such that
for $0<t<t_{0}$
\[
\frac{\widehat{C}_{v}{(t,v)}}{t}=\exp \left(  \log t\left(  A\left(
1,\frac{\log v}{\log t}\right)  -1\right)  \right)  A_{2}\left(  1,\frac{\log
v}{\log t}\right)  \frac{1}{v}\leq g\left(  {v}\right)  .
\]
Since $0<A_{2}(1,y)\leq1$ for any $y\geq0$, it follows that $A_{2}%
(1,\frac{\log v}{\log t})\frac{1}{v}\leq \frac{1}{v}$. Letting $f(t)=\exp
\left(  \log t\left(  A\left(  1,\frac{\log v}{\log t}\right)  -1\right)
\right)  $, by taking the derivative of $f(t)$ we obtain that $f(t)$ is a
nonincreasing function, which implies that $f(t)\leq f(0)=v^{A_{2}(1,0)}$ for
$t>0$. This leads to
\begin{equation}
\frac{\widehat{C}_{v}{(t,v)}}{t}\leq v^{A_{2}(1,0)-1}. \label{DCT}%
\end{equation}

Note that for any $0<\delta \leq1/2$, by (\ref{DCT}),
\[
\frac{\int_{0}^{\delta}y\widehat{C}_{v}\left(  \overline{F}(t),\overline
{F}\left(  ty\right)  \right)  dF(ty)}{\left(  \overline{F}\left(  t\right)
\right)  ^{A(1,1)}}\leq \frac{\int_{0}^{\delta}y\left(  \overline
{F}(ty)\right)  ^{A_{2}(1,0)-1}dF(ty)}{\left(  \overline{F}\left(  t\right)
\right)  ^{A(1,1)-1}}=\frac{\int_{0}^{\delta}yd\widetilde{F}(ty)}{\left(
\overline{F}\left(  t\right)  \right)  ^{A(1,1)-1}}.
\]
By Lemma 5.6 in Barbe and McCormick \cite{BM-2005}, if $A(1,1)<A_{2}(1,0)+1$,
then (\ref{A5}) holds. Thus, if $\alpha A_{1}(1,1)<1$, $\alpha(1-\beta)=\alpha
A_{2}(1,0)<1$ and $A(1,1)<A_{2}(1,0)+1$, which is the case $(\alpha,A)\in
C_{1}\cap C_{2}\cap C_{3}$, then all the conditions in Theorem \ref{TH1-P-D1}
are satisfied and the result in (i) follows.

Now we turn to the cases in Theorem \ref{Thm2}. The existence of function $g$
is verified by (\ref{DCT}). The verification of (\ref{A5}) also verifies
(\ref{cond}). Thus, to apply Theorem \ref{Thm2}, we are left to express
$\int_{0}^{1/2}((1-y)^{-\alpha}-1)\varphi(1,\overline{F}\left(  ty\right)
)dF(ty)$ explicitly. Note that%
\[
\int_{0}^{1/2}((1-y)^{-\alpha}-1)\varphi(1,\overline{F}\left(  ty\right)
)dF(ty)=\int_{0}^{1/2}((1-y)^{-\alpha}-1)d\widetilde{F}(ty).
\]
Since $\left(  \overline{F}(y)\right)  ^{A_{2}(1,0)}$ is regularly varying
with index $\alpha A_{2}(1,0)$, by Lemma 5.6 of Barbe and McCormick
\cite{BM-2005} if $\alpha A_{2}(1,0)<1$, then
\[
\lim_{t\rightarrow \infty}\frac{\int_{0}^{1/2}((1-y)^{-\alpha}-1)d\widetilde
{F}(ty)}{1-\widetilde{F}(t)}=I\left(  \alpha,\alpha A_{2}(1,0)\right)
<\infty.
\]
This is the case $(\alpha,A)\in C_{1}\backslash(C_{2}\cap C_{3})$, and by
applying Theorem \ref{Thm2} the result in (ii) follows. Lastly, by Lemma 2.4
of Mao and Hu \cite{Mao-Hu-2013} if $\alpha A_{2}(1,0)\geq1$ then
\[
\lim_{t\rightarrow \infty}\frac{\int_{0}^{1/2}((1-y)^{-\alpha}-1)d\widetilde
{F}(ty)}{t^{-1}\mu_{\widetilde{F}}(t)}=\alpha,
\]
where $\mu_{\widetilde{F}}(t)=\int_{0}^{t}xd\widetilde{F}(x)$. This is the
case $(\alpha,A)\in C_{1}^{c}$, and the result in (iii) follows by Theorem
\ref{Thm2}.
\end{proof}

\bigskip

Next, we present the asymptotic expansions for VaR.

\begin{theorem}
\label{EVCvar}Under the assumptions of Lemma \ref{EVClm} and further assume
that $\overline{F}(t)$ satisfies (\ref{2RV}), we have as $q\rightarrow1$, (i)
if $(\alpha,A)\in C_{1}\cap C_{2}\cap C_{3}$ and $\rho<-\alpha \left(
A(1,1)-1\right)  $, then%
\[
\mathrm{VaR}_{q}(Z)=2^{1/\alpha}\mathrm{VaR}_{q}(X)\left(  1+\frac
{2^{2-A(1,1)}\zeta_{1}}{\alpha}\left(  1-q\right)  ^{A(1,1)-1}(1+o(1))\right)
;
\]

(ii) if $(\alpha,A)\in C_{1}\backslash(C_{2}\cap C_{3})$ and $\rho<-\alpha
A_{2}(1,0)$, then
\[
\mathrm{VaR}_{q}(Z)=2^{1/\alpha}\mathrm{VaR}_{q}(X)\left(  1+\frac
{2^{2-A_{2}(1,0)}\zeta_{2}}{\alpha}\left(  1-q\right)  ^{A_{2}(1,0)}%
(1+o(1))\right)  .
\]

(iii) if $(\alpha,A)\in C_{1}^{c}$ and $\rho<-1$, then%
\[
\mathrm{VaR}_{q}(Z)=\mathrm{VaR}_{q}(X)\left(  2^{1/\alpha}+\mu_{\widetilde
{F}}\left(  F^{\leftarrow}(q)\right)  \left(  F^{\leftarrow}(q)\right)
^{-1}(1+o(1))\right)  ;
\]

(iv) if $(\alpha,A)\in C_{1}\cap C_{2}\cap C_{3}$ and $-\alpha \left(
A(1,1)-1\right)  <\rho<0$, or if $(\alpha,A)\in C_{1}\backslash(C_{2}\cap
C_{3})$ and $-\alpha A_{2}(1,0)<\rho<0$, or if $(\alpha,A)\in C_{1}^{c}$ and
$-1<\rho<0$, then%
\[
\mathrm{VaR}_{q}(Z)=2^{1/\alpha}\mathrm{VaR}_{q}(X)\left(  1+B\alpha
^{-1}\left(  2^{-\rho/\alpha}-1\right)  \left(  \mathrm{VaR}_{q}(X)\right)
^{\rho}(1+o(1))\right)  .
\]

\end{theorem}

\begin{proof}
(i) From Theorem \ref{EVC} (i), note that if $t^{\rho}=o\left(  \left(
\overline{F}(t)\right)  ^{A(1,1)-1}\right)  $, i.e. $\rho<-\alpha \left(
A(1,1)-1\right)  $, then for any $x>0$,
\[
\frac{\overline{G}(tx)}{\overline{F}(t)}=x^{-\alpha}\left(  2+\zeta
_{1}x^{-\alpha(A(1,1)-1)}\left(  \overline{F}(t)\right)  ^{A(1,1)-1}%
(1+o(1))\right)  .
\]
The desired result follows similarly as in the proof of Theorem \ref{varind}.

(ii) From Theorem \ref{EVC} (ii), note that if $t^{\rho}=o\left(  \left(
\overline{F}(t)\right)  ^{A_{2}(1,0)}\right)  $ i.e. $\rho<-\alpha A_{2}%
(1,0)$, then for any $x>0$,%
\[
\frac{\overline{G}(tx)}{\overline{F}(t)}=x^{-\alpha}\left(  2+\zeta
_{2}x^{-\alpha A_{2}(1,0)}\left(  \overline{F}(t)\right)  ^{A_{2}%
(1,0)}(1+o(1))\right)  .
\]
The rest follows the proof of Theorem \ref{varind}.

(iii) and (iv) follow similarly.\bigskip
\end{proof}

Now we show a numerical example to illustrate the asymptotic expansions for
tail probabilities and for VaR under the Gumbel copula, which is the only
copula belongs to both extreme value and the Archimedean copulas; see Genest
and Rivest \cite{GR1989}. First we verify that $A(u,v)=\left(  u^{\phi
}+v^{\phi}\right)  ^{1/\phi}$ with $\phi \geq1$ satisfies condition
(\ref{evcond}). With some rewriting, we obtain for all $x\geq0$ and $0<t\leq
t_{0}$, and $0\leq u<\delta$,
\begin{align*}
\frac{A_{2}\left(  1,\frac{x}{1+\frac{\log(1+u)}{\log t}}\right)  }%
{A_{2}(1,x)}  &  =\left(  \frac{\left(  1+\frac{\log(1+u)}{\log t}\right)
^{\phi}+x^{\phi}}{1+x^{\phi}}\right)  ^{1/\phi-1}\\
&  =\left(  1+\frac{1}{1+x^{\phi}}\left(  \left(  1+\frac{\log(1+u)}{\log
t}\right)  ^{\phi}-1\right)  \right)  ^{1/\phi-1}\\
&  \leq \left(  1+\frac{\log(1+u)}{\log t}\right)  ^{1-\phi}\\
&  \leq \left(  1+\frac{\log(1+u)}{\log t_{0}}\right)  ^{1-\phi}\\
&  \sim(1+u)^{(1-\phi)/\log t_{0}},
\end{align*}
where the third step is due that $\frac{1}{1+x^{\phi}}\leq1$ and $\left(
1+\frac{\log(1+u)}{\log t}\right)  ^{\phi}<1$ and in the last step we used the
fact $\log(1+u)\sim u$ and $(1+u)^{c}\sim1+cu$ for $0\leq u\leq \delta$. This
shows that the Gumbel copula satisfies the conditions of Lemma \ref{EVClm} and
hence Theorems \ref{EVC} and \ref{EVCvar} can be applied.

To be more specific, in the numerical example, we assume that the survival
copula for $(X,Y)$ is a bivariate Gumbel copula so that the joint survival
probability of $(X,Y)$ is given by%
\[
\Pr \left(  X>x,Y>y\right)  ={\exp}\left(  -\left(  (-\log \overline{{F}}%
{(x))}^{\phi}+(-\log{\overline{{F}}{(y))}}^{\phi}\right)  ^{1/\phi}\right)  .
\]
The marginal distribution is assumed to be a Pareto distribution with
distribution function
\[
F(x)=1-\left(  \frac{\theta}{x+\theta}\right)  ^{\alpha},\qquad x>0
\]
with $\alpha>0$ and $\theta>0$. We can verify that
\[
\overline{{F}}{(x)=\theta}^{\alpha}x^{-\alpha}\left(  1-\alpha \theta
x^{-1}(1+o(1))\right)  ,
\]
which shows (\ref{2RV}) is satisfied with $\rho=-1$, $A={\theta}^{\alpha}$ and
$B=-\alpha \theta$.

To investigate the performance of the asymptotic expansion for tail
probabilities, on one hand, we calculate the asymptotic expansions of the tail
probabilities $p_{t}:=\Pr \left(  X+Y>t\right)  $ for large $t$ and VaR of
$Z=X+Y$, $\mathrm{VaR}_{q}(Z)$ for $q$ close to $1$ when the survival copula
for $(X,Y)$ follows a Gumbel copula with Pareto marginals by applying Theorems
\ref{EVC} and \ref{EVCvar}. On the other hand,{ we simulate a random sample
$\{(X_{i},Y_{i}),\  \,i=1,\ldots,n\}$ of $(X,Y)$ with size $n=100,000$ from the
same Gumbel copula and Pareto marginals and estimate the tail probability
$p_{t}$ by
\[
\widehat{p}_{t}:=\frac{1}{n}\sum_{i=1}^{n}1_{\{X_{i}+Y_{i}>t\}},
\]
and estimate VaR by}%
\[
\widehat{\mathrm{VaR}}_{q}(Z):=Z_{\lfloor nq\rfloor,n}%
\]
where $Z_{1,n}\leq Z_{2,n}\leq \cdots Z_{n,n}$ are order statistics. Let the
Pareto marginals have varying tail index $\alpha$ and fixed scale parameter of
$1$. The Gumbel copula has a varying parameter $\phi=1$ and $10$. Here $\phi$
controls the dependence. When $\phi=1$ the Gumbel copula corresponds to the
independence copula. In this case, one can verify that Theorem \ref{EVC}
reduces to Proposition \ref{indp} and Theorem \ref{EVCvar} reduces to Theorem
\ref{varind}. We choose $\alpha=0.8$ and $2$ respectively and show the values
of $\Pr(X+Y>t)$ and $\mathrm{VaR}_{q}(Z)$ from both the asymptotic expansions
and simulations in Figure \ref{f1}. When $\phi=10$, it corresponds to a
stronger dependent case. Again we choose $\alpha=0.8$ and $2$ respectively,
and the corresponding results are plotted in Figure \ref{f2}. Overall, the
asymptotic values are very close to the simulated ones, especially for the
heavier tailed case.

\begin{figure}[ptb]
\begin{subfigure}{0.5\textwidth}
	 \centering
		\includegraphics[width=0.7\textwidth]{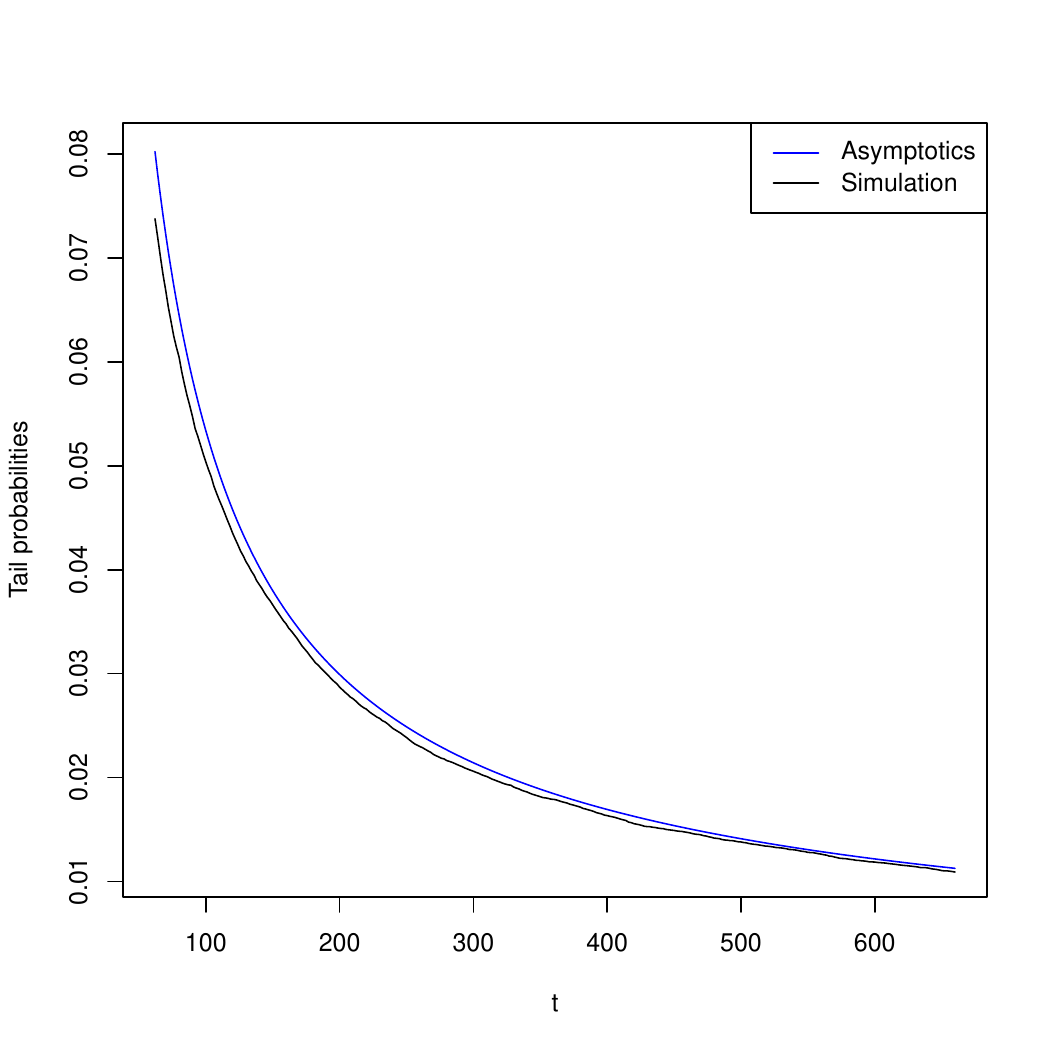}
		\caption{Tail probability with $\alpha=0.8$}
		\label{fig1}
	\end{subfigure}
\begin{subfigure}{0.5\textwidth}	
	 \centering	
		\includegraphics[width=0.7\textwidth]{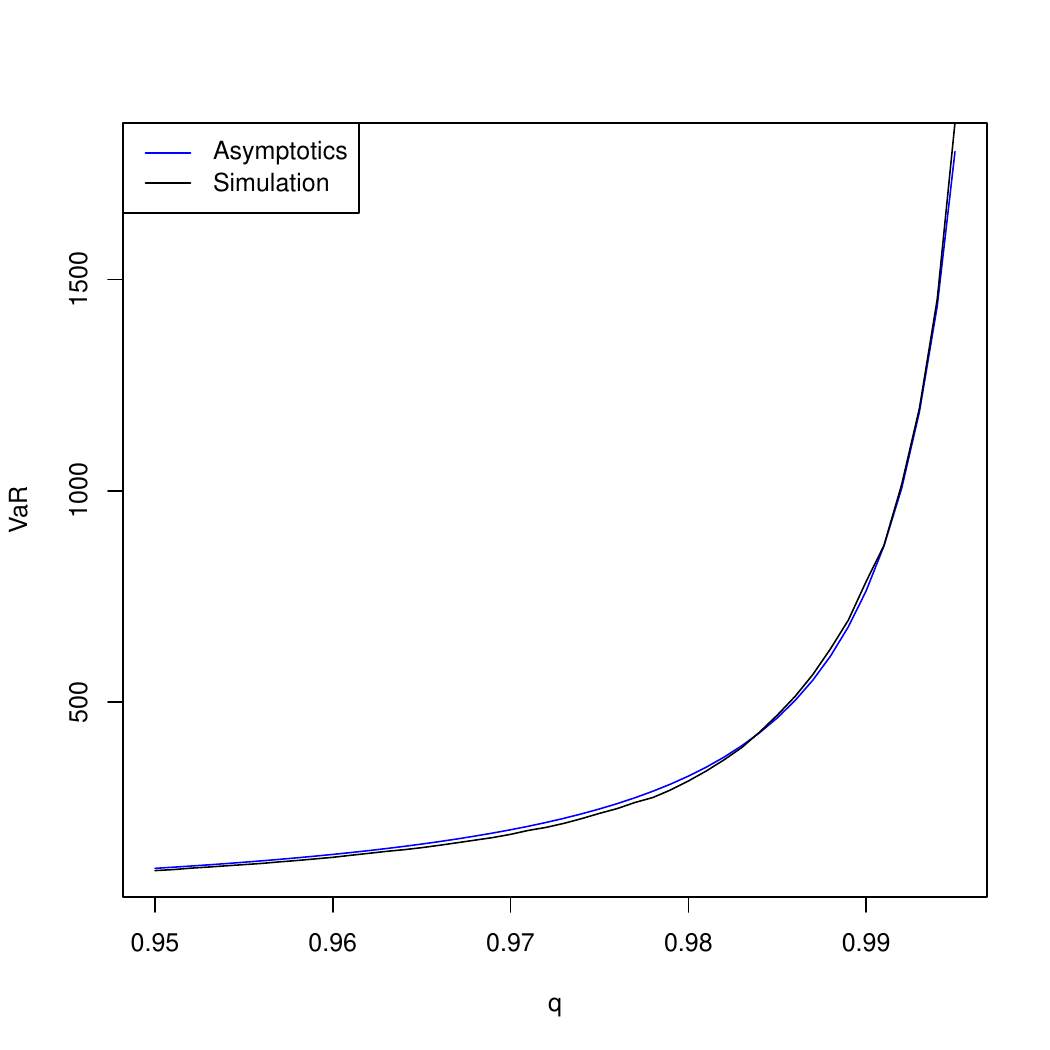}
	\caption{VaR with $\alpha=0.8$ }
		\label{fig2}
	\end{subfigure}
\begin{subfigure}{0.5\textwidth}
	\centering
	\includegraphics[width=0.7\textwidth]{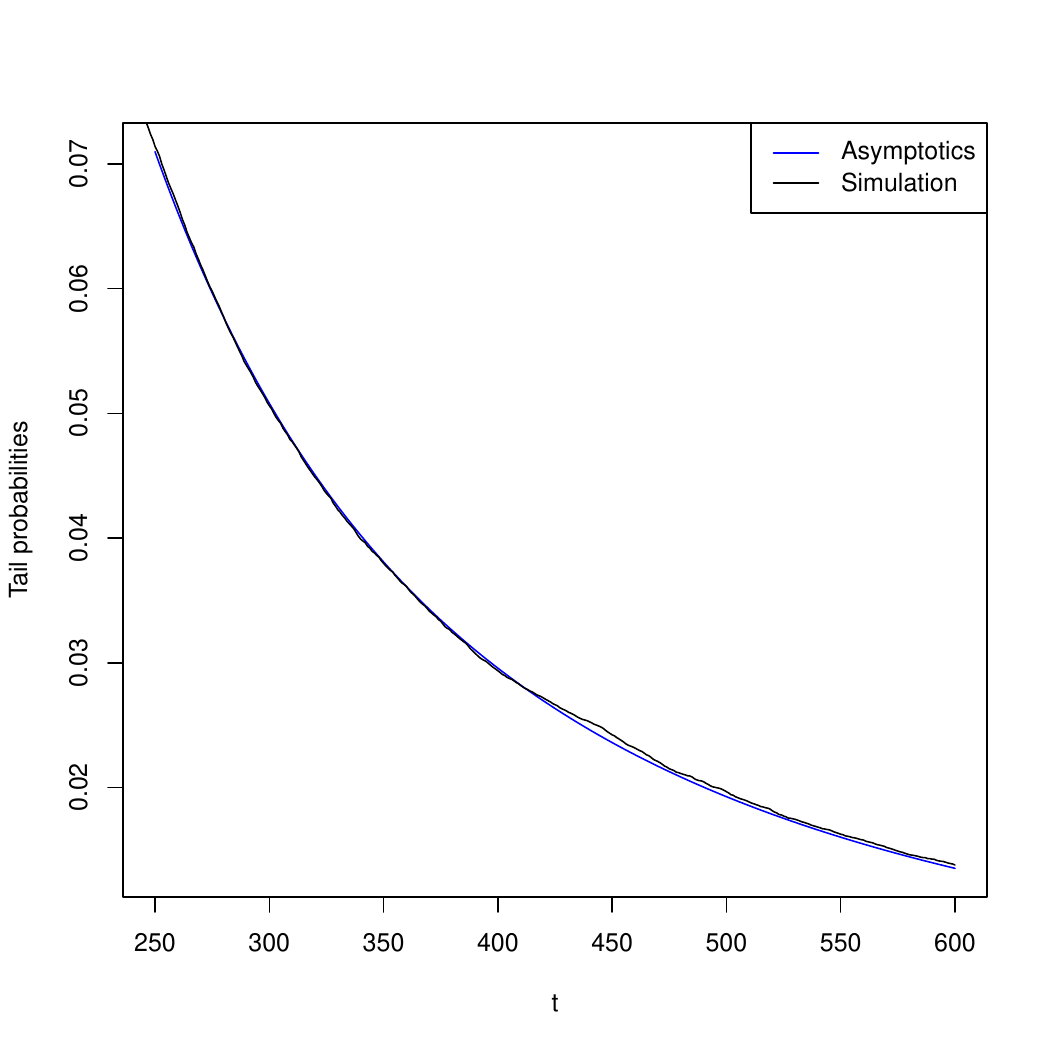}
	\caption{Tail probability with $\alpha=2$ }
	\label{fig3}
\end{subfigure}
\begin{subfigure}{0.5\textwidth}	
	\centering	
	\includegraphics[width=0.7\textwidth]{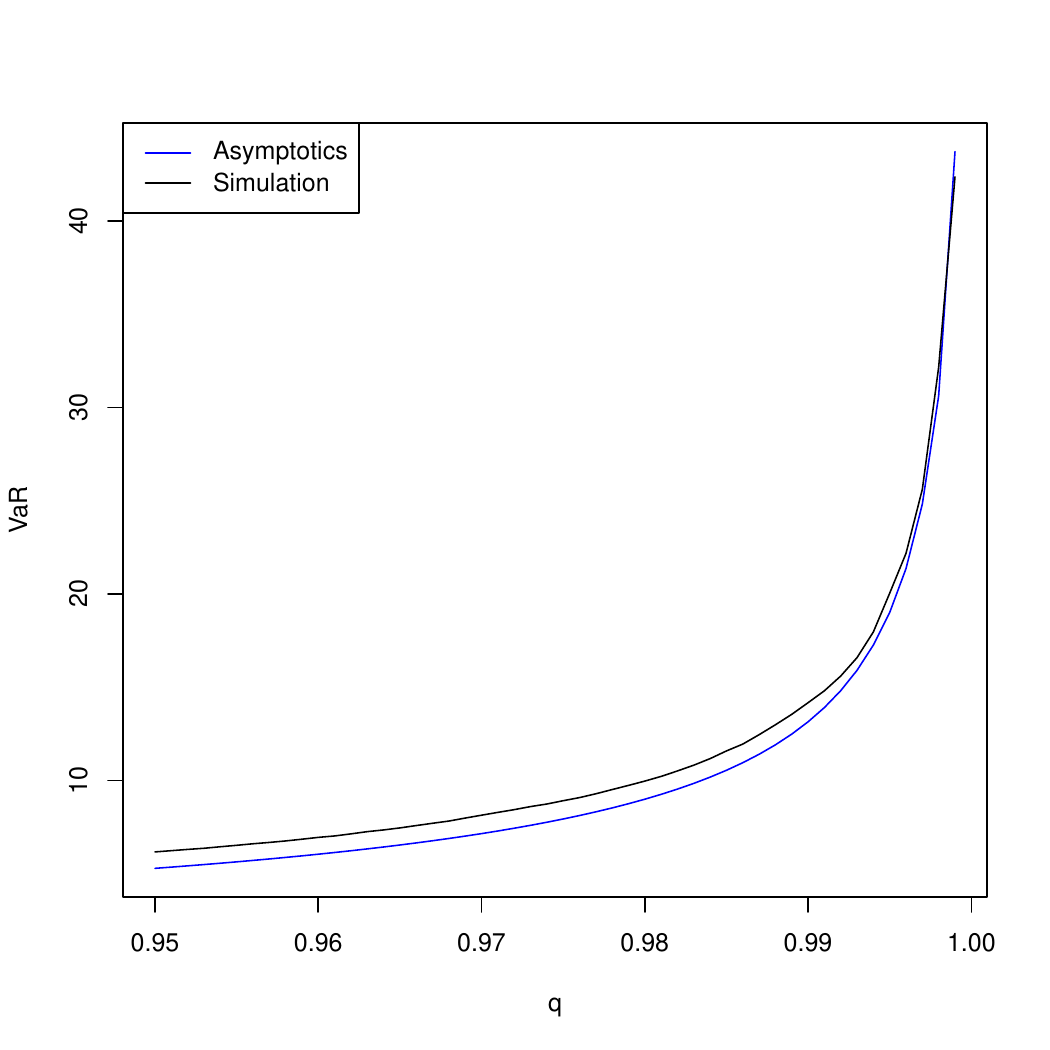}
	\caption{VaR with $\alpha=2$ }
	\label{fig4}
\end{subfigure}
\caption{Tail probabilities and VaR under Gumbel copula with $\phi=1$.}%
\label{f1}%
\end{figure}

\begin{figure}[ptb]
\begin{subfigure}{0.5\textwidth}
	\centering
	\includegraphics[width=0.7\textwidth]{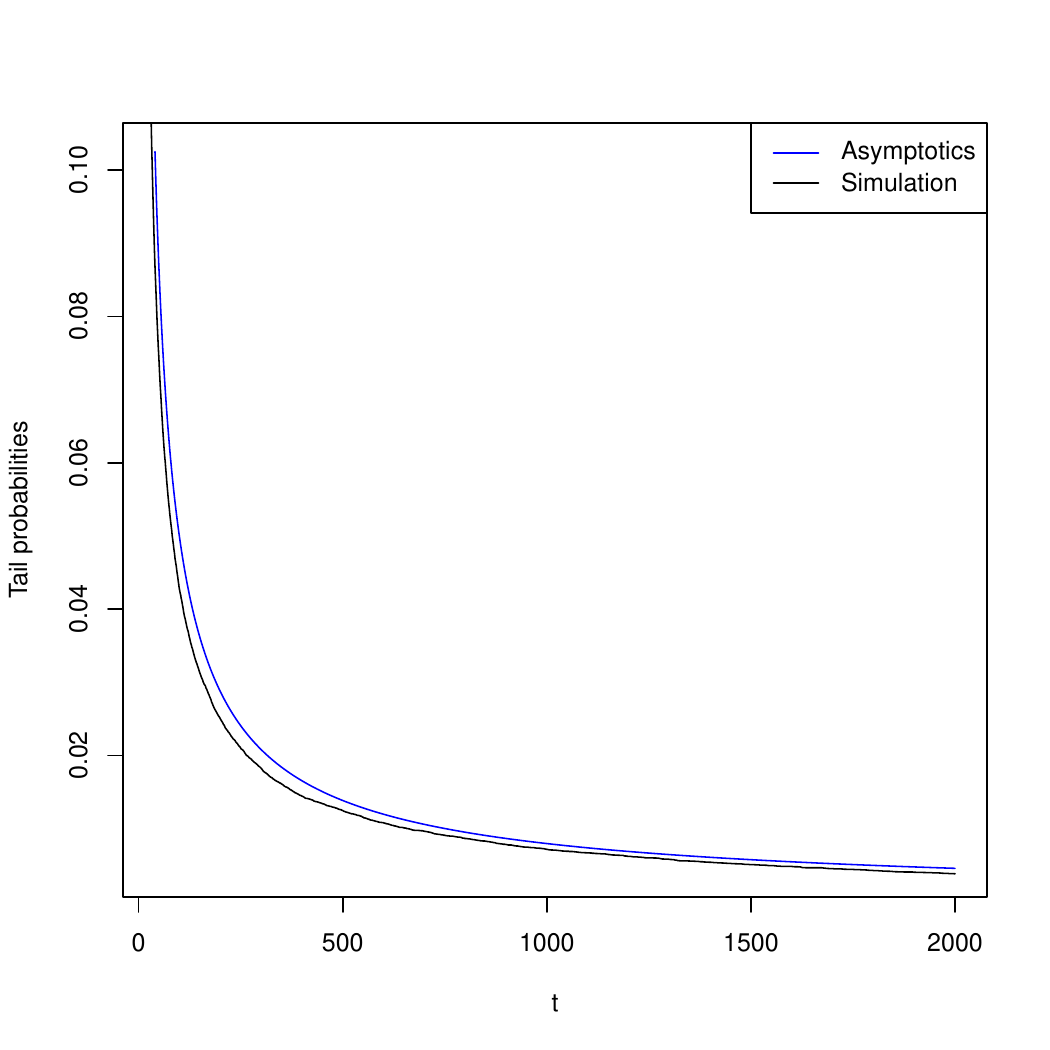}
	\caption{Tail probability with $\alpha=0.8$ }
	\label{fig5}
\end{subfigure}
\begin{subfigure}{0.5\textwidth}	
	\centering	
	\includegraphics[width=0.7\textwidth]{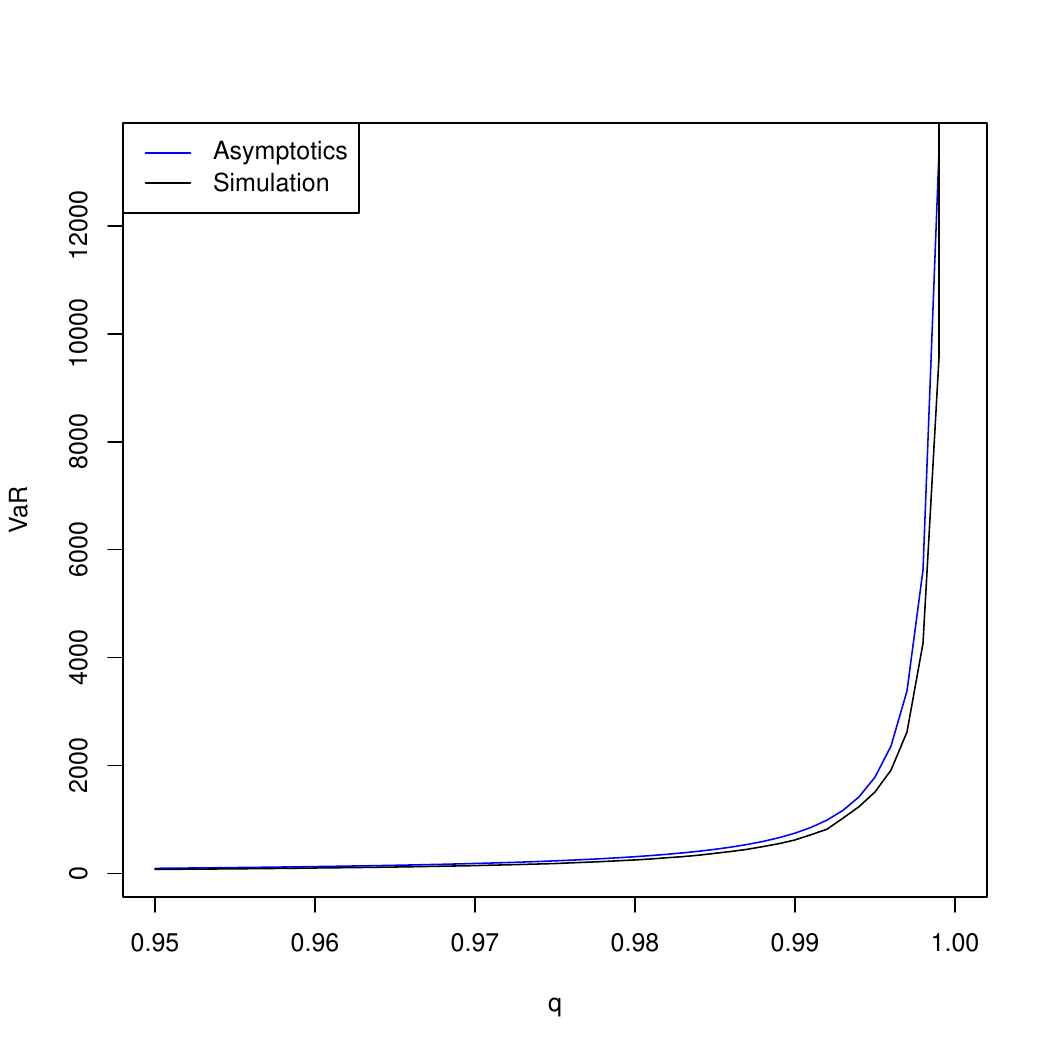}
	\caption{VaR with $\alpha=0.8$ }
	\label{fig6}
\end{subfigure}
\begin{subfigure}{0.5\textwidth}
	\centering
	\includegraphics[width=0.7\textwidth]{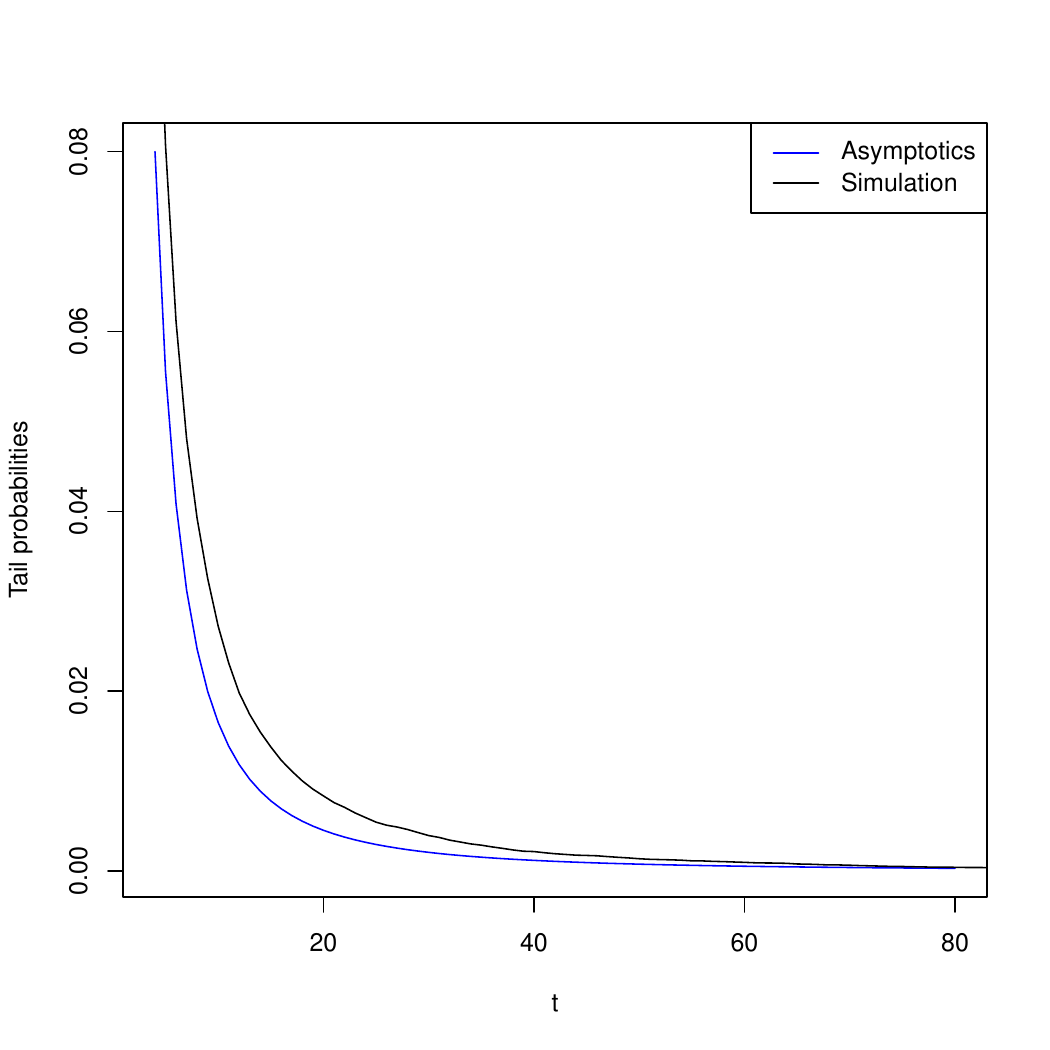}
	\caption{Tail probability with $\alpha=2$}
	\label{fig7}
\end{subfigure}
\begin{subfigure}{0.5\textwidth}	
	\centering	
	\includegraphics[width=0.7\textwidth]{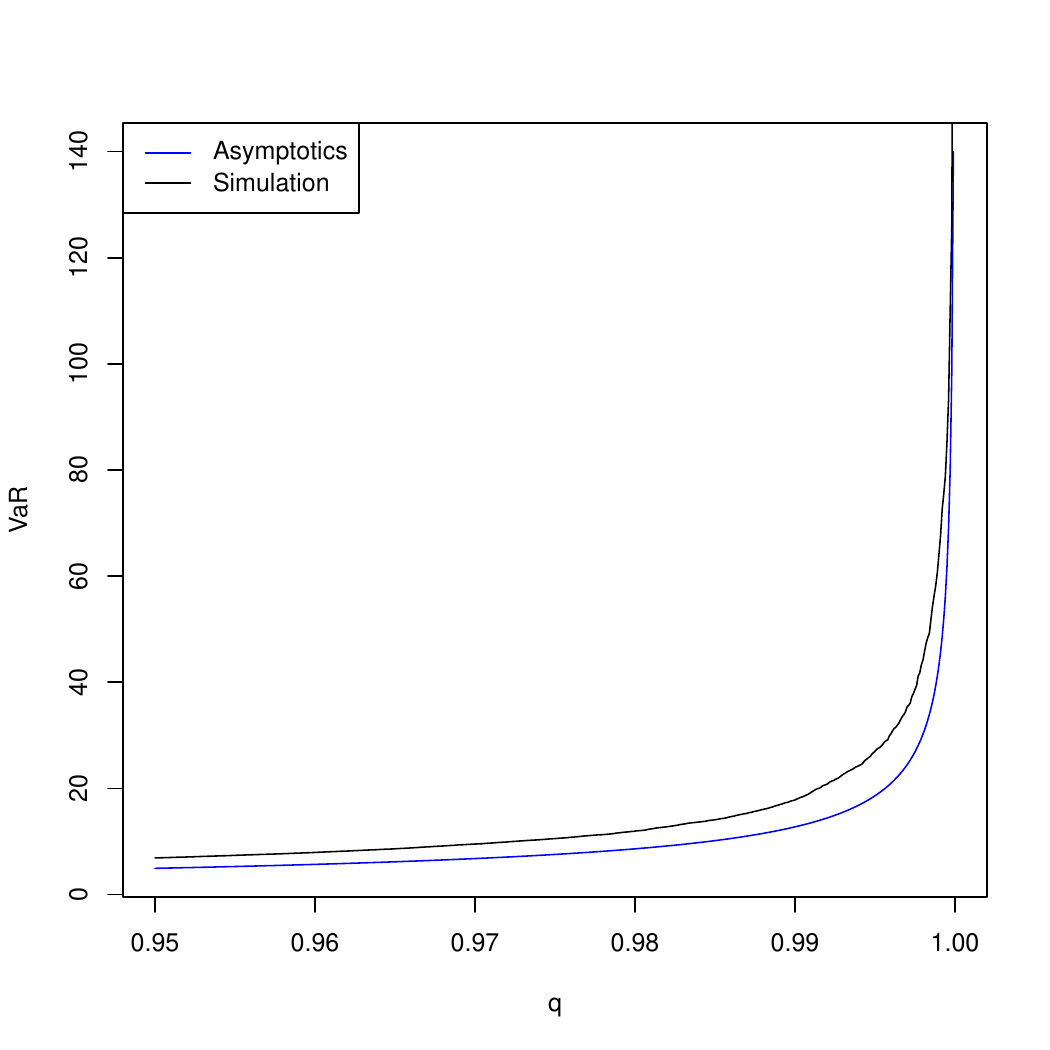}
	\caption{VaR with $\alpha=2$ }
	\label{fig8}
\end{subfigure}
\caption{Tail probabilities and VaR under Gumbel copula with $\phi=10$.}%
\label{f2}%
\end{figure}\bigskip

\section{Proofs\label{proof}}

In this section, we introduce a few lemmas that are the key steps in the
proofs of main results, which are the important characterizations of survival
copulas as well.

For notational simplicity, we define for any $0<\delta<1/2$,
\begin{equation}
D(\delta,t)=\int_{\delta}^{1/2}\left(  \tau_{v}\left(  \left(  1-y\right)
^{-\alpha},y^{-\alpha}\right)  -\tau_{v}\left(  1,y^{-\alpha}\right)  \right)
dF(ty), \label{D_delta_t}%
\end{equation}
{where} $\tau_{v}(u,v)=\frac{\partial}{\partial v}\tau(u,v)$, and denote%
\begin{equation}
\eta(\delta)=\alpha \int_{\delta}^{1/2}\left(  \tau_{v}\left(  \left(
1-y\right)  ^{-\alpha},y^{-\alpha}\right)  -\tau_{v}\left(  1,y^{-\alpha
}\right)  \right)  y^{-\alpha-1}dy. \label{eta}%
\end{equation}
When $\lim_{\delta \downarrow0}\eta(\delta)$ exists, for simplicity we denote
$\lim_{\delta \downarrow0}\eta(\delta)$ by $\eta$.

Starting from the point-wise convergence in (\ref{assumption on tail order}),
{the following two lemmas extend the locally uniformly convergence of
univariate RV function to the bivariate RV function.}

\begin{lemma}
\label{unifA2}Under Assumption \ref{A2}, for any compact set $B\subset
\lbrack0,\infty)^{2}$,%
\[
\lim_{t\downarrow0}\sup_{(u,v)\in B}\left \vert \frac{\widehat{C}%
(ut,vt)}{t^{\kappa}\ell(t)}-\tau(u,v)\right \vert =0.
\]

\end{lemma}

\begin{proof}
Fix a compact set $B$. We prove the lemma by contradiction. Suppose the
assertion in the lemma does not hold. Without loss of generality, there exist
some $\varepsilon_{0}>0$ and sequences $\{t_{n}\}$, $\{u_{n}\}$ and
$\{v_{n}\}$ such that
\begin{equation}
\frac{\widehat{C}\left(  u_{n}t_{n},v_{n}t_{n}\right)  }{t_{n}^{\kappa}%
\ell(t_{n})}-\tau(u_{n},v_{n})>\varepsilon_{0}, \label{ineq1}%
\end{equation}
where $t_{n}$ converges to $0$ as $n\rightarrow \infty$. Since the sequences
$\{u_{n}\}$ and $\{v_{n}\}$ are bounded, by the Bolzano--Weierstrass Theorem
there exists a convergent subsequences $\{u_{n_{k}}\}$ and $\{v_{n_{k}}\}$
with $u_{n_{k}}\rightarrow u_{0}$, $v_{n_{k}}\rightarrow v_{0}$ as
$n_{k}\rightarrow \infty$. We choose $0<\varepsilon<\varepsilon_{0}/2$ and
$\delta>0$ such that when $n_{k}$ is large $u_{0}-\delta<u_{n_{k}}%
<u_{0}+\delta$, $v_{0}-\delta<v_{n_{k}}<v_{0}+\delta$ (if $u_{0}=0$ or
$v_{0}=0$, then we replace the previous inequality by $0<u_{n_{k}}<\delta$ or
$0<v_{n_{k}}<\delta$), and
\[
\tau(u_{0}+\delta,v_{0}+\delta)-\tau(u_{0}-\delta,v_{0}-\delta)<\varepsilon.
\]
Note that $\widehat{C}(u,v)$ and $\tau(u,v)$ are both nondecreasing in $u$ and
$v$. Then using (\ref{ineq1}), and the monotonicity of $\widehat{C}$ and
$\tau$, we have%
\begin{align*}
\varepsilon_{0}  &  <\frac{\widehat{C}\left(  u_{n_{k}}t_{n_{k}},v_{n_{k}%
}t_{n_{k}}\right)  }{t_{n_{k}}^{\kappa}\ell(t_{n_{k}})}-\tau(u_{n_{k}%
},v_{n_{k}})\\
&  \leq \frac{\widehat{C}\left(  \left(  u_{0}+\delta \right)  t_{n_{k}}%
,(v_{0}+\delta)t_{n_{k}}\right)  }{t_{n_{k}}^{\kappa}\ell(t_{n_{k}})}%
-\tau(u_{0}-\delta,v_{0}-\delta)\\
&  \leq \frac{\widehat{C}\left(  \left(  u_{0}+\delta \right)  t_{n_{k}}%
,(v_{0}+\delta)t_{n_{k}}\right)  }{t_{n_{k}}^{\kappa}\ell(t_{n_{k}})}%
-\tau(u_{0}+\delta,v_{0}+\delta)+\varepsilon.
\end{align*}
Thus,%
\[
\lim_{n_{k}\rightarrow \infty}\left(  \frac{\widehat{C}\left(  \left(
u_{0}+\delta \right)  t_{n_{k}},(v_{0}+\delta)t_{n_{k}}\right)  }{t_{n_{k}%
}^{\kappa}\ell(t_{n_{k}})}-\tau(u_{0}+\delta,v_{0}+\delta)\right)
>\frac{\varepsilon_{0}}{2},
\]
which contradicts to (\ref{assumption on tail order}). This completes the
proof.\bigskip
\end{proof}

\begin{lemma}
{\label{lemma: copula convergence}Under Assumptions \ref{A1} and\textbf{\ }%
}\ref{A2}{, for any compact set \footnote{A compact set $B\subset \left(
0,\infty \right)  ^{2}$ is bounded away from $0$.} $B\subset(0,\infty)^{2}$, we
have}%

\[
\lim_{t\rightarrow \infty}\sup_{(x,y)\in B}\left \vert \frac{\widehat{C}\left(
\overline{F}(xt),\overline{F}(yt)\right)  }{\left(  \overline{F}(t)\right)
^{\kappa}\ell(\overline{F}(t))}-\tau(x^{-\alpha},y^{-\alpha})\right \vert =0.
\]

\end{lemma}

\begin{proof}
Fix the compact set $B$. Since $\overline{F}\in \text{RV}_{-\alpha}$ with
$\alpha>0$, by Potter's bound (e.g. Proposition B.1.9.5 of de Haan and
Ferreira \cite{de Haan-Ferreira-2007}), for any $\varepsilon,\delta>0$ there
exists $t_{0}=t_{0}(\varepsilon,\delta)>0$ such that for any $t>t_{0}$ and
$(x,y)\in B$ we have
\begin{equation}
(1-\varepsilon)\left(  x^{-\alpha+\delta}\wedge x^{-\alpha-\delta}\right)
\leq \frac{\overline{F}(xt)}{\overline{F}(t)}\leq(1+\varepsilon)\left(
x^{-\alpha+\delta}\vee x^{-\alpha-\delta}\right)  , \label{cgw-06-26-01}%
\end{equation}
and the same inequality by replacing $x$ with $y$. By the monotonicity of the
survival copula, it follows from Lemma \ref{unifA2} and the second inequality
in (\ref{cgw-06-26-01}) that
\begin{align*}
&  \lim_{t\rightarrow \infty}\sup_{(x,y)\in B}\left(  \frac{\widehat
{C}(\overline{F}(xt),\overline{F}(yt))}{\left(  \overline{F}(t)\right)
^{\kappa}\ell(\overline{F}(t))}-\tau(x^{-\alpha},y^{-\alpha})\right) \\
&  \leq \lim_{\varepsilon,\delta \rightarrow0}\lim_{t\rightarrow \infty}%
\sup_{(x,y)\in B}\left \vert \frac{\widehat{C}\left(  (1+\varepsilon)\left(
x^{-\alpha+\delta}\vee x^{-\alpha-\delta}\right)  \overline{F}%
(t),(1+\varepsilon)\left(  y^{-\alpha+\delta}\vee y^{-\alpha-\delta}\right)
\overline{F}(t)\right)  }{\left(  \overline{F}(t)\right)  ^{\kappa}\ell \left(
\overline{F}(t)\right)  }\right. \\
&  \left.  -\tau \left(  (1+\varepsilon)\left(  x^{-\alpha+\delta}\vee
x^{-\alpha-\delta}\right)  ,(1+\varepsilon)\left(  y^{-\alpha+\delta}\vee
y^{-\alpha-\delta}\right)  \right)  \right \vert \\
&  +\lim_{\varepsilon,\delta \rightarrow0}\sup_{(x,y)\in B}\left \vert
\tau \left(  (1+\varepsilon)\left(  x^{-\alpha+\delta}\vee x^{-\alpha-\delta
}\right)  ,(1+\varepsilon)\left(  y^{-\alpha+\delta}\vee y^{-\alpha-\delta
}\right)  \right)  -\tau \left(  x^{-\alpha},y^{-\alpha}\right)  \right \vert \\
&  =0,
\end{align*}
where the last step is due to (\ref{assumption on tail order}) and
that{\ $\tau(\cdot)$ is uniform continuous on the compact set $B$.} Using the
first inequality in \eqref{cgw-06-26-01}, we can similarly prove
\[
\lim_{t\rightarrow \infty}\sup_{(x,y)\in B}\left(  \tau(x^{-\alpha},y^{-\alpha
})-\frac{\widehat{C}(\overline{F}(xt),\overline{F}(yt))}{(\overline
{F}(t))^{\kappa}\ell(\overline{F}(t))}\right)  =0.
\]
Thus, the desired result follows.\bigskip
\end{proof}

{The differentiability of $\widehat{C}$ and }$\tau(u,v)$ in {Assumption
}\ref{A2} {is needed to study the convergence of partial derivatives of the
survival copula as in the next lemma.}

\begin{lemma}
{\label{lemma: partial derivative of copula regular varying}Under Assumption
}\ref{A2}, for any compact set $B\subset(0,\infty)^{2}$, it holds that{\
\[
\lim_{t\downarrow0}\sup_{(u,v)\in B}\left \vert \frac{\widehat{C}_{v}%
(ut,vt)}{t^{\kappa-1}\ell(t)}-\tau_{v}(u,v)\right \vert =0.
\]
}
\end{lemma}

\begin{proof}
Note that $\widehat{C}_{v}(u,v)$ and $\tau_{v}(u,v)$ are nondecreasing in $u$
and nonincreasing in $v$. Because of such monotonicity, we only need to show
the pointwise convergence
\begin{equation}
\lim_{t\downarrow0}\frac{\widehat{C}_{v}(ut,vt)}{t^{\kappa-1}\ell(t)}=\tau
_{v}(u,v). \label{pwc}%
\end{equation}
The locally uniform convergence can be showed similarly as Lemma \ref{unifA2}.

Since both $\widehat{C}(u,v)$ and $\tau(u,v)$ are differentiable, for any
$u,v\geq0$ by the mean value theorem, we have for any $\varepsilon>0$ there
exists $\xi_{1},\xi_{2}\in \lbrack v,v+\varepsilon]$ such that%
\[
\widehat{C}(ut,vt+\varepsilon t)-\widehat{C}(ut,vt)=\widehat{C}_{v}(ut,\xi
_{1}t)\varepsilon t
\]
and%
\[
\tau(u,v+\varepsilon)-\tau(u,v)=\tau_{v}(u,\xi_{2})\varepsilon.
\]
Then as $t\downarrow0$, we have%
\[
\frac{\widehat{C}(ut,vt+\varepsilon t)-\widehat{C}(ut,vt)}{t^{\kappa}\ell
(t)}\rightarrow \tau(u,v+\varepsilon)-\tau(u,v),
\]
or equivalently, dividing $\varepsilon$ on both sides leads%
\begin{align*}
\frac{\left(  \widehat{C}(ut,vt+\varepsilon t)-\widehat{C}(ut,vt)\right)
/(\varepsilon t)}{t^{\kappa-1}\ell(t)}  &  =\frac{\widehat{C}_{v}(ut,\xi
_{1}t)}{t^{\kappa-1}\ell(t)}\\
&  \rightarrow \frac{\tau(u,v+\varepsilon)-\tau(u,v)}{\varepsilon}=\tau
_{v}(u,\xi_{2}).
\end{align*}
By the arbitrariness of $\varepsilon$, the pointwise convergence (\ref{pwc})
holds.\bigskip
\end{proof}

By noting that $\widehat{C}_{v}\left(  u,v\right)  $ is nondecreasing in $u$
and nonincreasing in $v$, similar to the proof of Lemma
\ref{lemma: copula convergence} we obtain the following corollary.

\begin{corollary}
{\label{coro: partial derivative of copula regular varying-1}Under the
conditions of Lemma \ref{lemma: copula convergence}, it holds for any compact
set }$B\subset(0,\infty)^{2}$ that%
\[
\lim_{t\rightarrow \infty}\sup_{(x,y)\in B}\left \vert \frac{\widehat{C}%
_{v}(\overline{F}(xt),\overline{F}(yt))}{(\overline{F}(t))^{\kappa-1}%
\ell(\overline{F}(t))}-\tau_{v}(x^{-\alpha},y^{-\alpha})\right \vert =0.
\]

\end{corollary}

Next we show the homogeneity of $\tau_{v}$ and $\varphi$.

\begin{lemma}
\label{homo}Under Assumption \ref{A2}, $\tau_{u}(u,v)$ and $\tau_{v}(u,v)$ are
homogenous functions of degree $\kappa-1$. Under Assumption \ref{A3},
$\varphi(u,v)=u^{\theta}\varphi(1,v)$.
\end{lemma}

\begin{proof}
It suffices to show the result for $\tau_{v}(u,v)$ since the rest follows in
the same way. For any $s,u>0$ and $0<v\leq1$, it follows from the definition
that
\begin{align*}
\tau_{v}(su,sv)  &  =\lim_{t\downarrow0}\frac{\widehat{C}_{v}(sut,svt)}%
{t^{\kappa-1}\ell(t)}\\
&  =\lim_{t\downarrow0}\frac{\widehat{C}_{v}(sut,svt)}{(st)^{\kappa-1}%
\ell(st)}\frac{(st)^{\kappa-1}\ell(st)}{t^{\kappa-1}\ell(t)}\\
&  =s^{\kappa-1}\tau_{v}(u,v),
\end{align*}
by which the proof is complete.\bigskip
\end{proof}

The following lemma and remark explain the range of $\theta$ in Assumption
\ref{A4}.

\begin{lemma}
\label{Taylor}Assume that $\widehat{C}$ is twice differentiable with
$\widehat{C}_{vu}\left(  u,v\right)  =\frac{\partial}{\partial u}\widehat
{C}_{v}\left(  u,v\right)  $ satisfying $\widehat{C}_{vu}\left(  0,v\right)
\not =0$ for any $0\leq v\leq1$. Then for $u>0$ and $0<v<1$, it holds that
\[
\lim_{t\downarrow0}\frac{\widehat{C}_{v}\left(  ut,v\right)  }{t}=u\widehat
{C}_{vu}\left(  0,v\right)  .
\]

\end{lemma}

\begin{proof}
For some $u$ close to $0$, applying Taylor's theorem yields
\begin{equation}
\widehat{C}_{v}\left(  u,v\right)  =\widehat{C}_{v}\left(  0,v\right)
+\widehat{C}_{vu}\left(  0,v\right)  u+o\left(  u\right)  . \label{expansion}%
\end{equation}
Noting that $\widehat{C}_{v}\left(  u,v\right)  =\Pr \left(  X>F^{\leftarrow
}(1-u)|Y=F^{\leftarrow}(1-v)\right)  $, which yields $\widehat{C}_{v}\left(
0,v\right)  =0$, we have that
\[
\lim_{u\downarrow0}\frac{\widehat{C}_{v}\left(  u,v\right)  }{u}=\widehat
{C}_{vu}\left(  0,v\right)  .
\]
This ends the proof.\bigskip
\end{proof}

\begin{remark}
\label{remark}The parameter $\theta$ in Assumption \ref{A4} is at least $1$,
which can be seen from the proof of Lemma \ref{Taylor}. Since $\widehat{C}$ is
twice differentiable, the second term $\widehat{C}_{vu}\left(  0,v\right)  $
on the right-hand side of (\ref{expansion}) is well defined. This shows that
if $\widehat{C}_{vu}\left(  0,v\right)  \not =0$, then $\theta=1$; otherwise,
$\theta>1$.

\end{remark}

The following lemma presents an asymptotic expansion for $D(\delta,t)$ defined
in \eqref{D_delta_t}.

\begin{lemma}
\label{D1D2}Under Assumptions \ref{A1} and \ref{A2}, for $\delta \leq y\leq
1/2$, if $\mu \{y\in \lbrack \delta,\frac{1}{2}]:\tau_{v}\left(  \left(
1-y\right)  ^{-\alpha},y^{-\alpha}\right)  >\tau_{v}\left(  1,y^{-\alpha
}\right)  \}>0$ where $\mu$ is the Lebesgue measure, then%
\[
\lim_{t\rightarrow \infty}\frac{D(\delta,t)}{\overline{F}\left(  t\right)
}=\eta(\delta),
\]
where $\eta(\delta)$ is given by (\ref{eta}).
\end{lemma}

\begin{proof}
By Corollary 2.2.6. of Nelson \cite{Nelsen-2006}, $\widehat{C}_{v}\left(
u,v\right)  $ is nondecreasing in $u$. As a result, $\tau_{v}\left(
u,v\right)  $ is nondecreasing in $u$. Since $\left(  1-y\right)  ^{-\alpha
}>1$ for $\delta \leq y\leq1/2$, it follows that $\tau_{v}\left(  \left(
1-y\right)  ^{-\alpha},y^{-\alpha}\right)  \geq \tau_{v}\left(  1,y^{-\alpha
}\right)  $.

If $\mu \{y\in \lbrack \delta,\frac{1}{2}]:\tau_{v}\left(  \left(  1-y\right)
^{-\alpha},y^{-\alpha}\right)  >\tau_{v}\left(  1,y^{-\alpha}\right)  \}>0$,
then $D(\delta,t)>0$. Since $\tau(u,v)$ is differentiable on $[0,\infty)^{2}$,
for any $0<\delta<1/2$ it leads to
\[
\int_{\delta}^{1/2}\left(  \tau_{v}\left(  \left(  1-y\right)  ^{-\alpha
},y^{-\alpha}\right)  -\tau_{v}\left(  1,y^{-\alpha}\right)  \right)
y^{-\alpha-1}dy<\infty.
\]
Note that the measure $\left(  \overline{F}(t)\right)  ^{-1}dF(ty)$ converges
vaguely on $(0,\infty)$ to $\alpha y^{-\alpha-1}dy$. Thus,
\[
D(\delta,t)\sim \frac{\alpha}{\overline{F}(t)}\int_{\delta}^{1/2}\left(
\tau_{v}\left(  \left(  1-y\right)  ^{-\alpha},y^{-\alpha}\right)  -\tau
_{v}\left(  1,y^{-\alpha}\right)  \right)  y^{-\alpha-1}dy.
\]
This ends the proof.\bigskip
\end{proof}

As one can see from the proof of the above Lemma \ref{D1D2}, if $\mu
\{y\in \lbrack \delta,\frac{1}{2}]:\tau_{v}\left(  \left(  1-y\right)
^{-\alpha},y^{-\alpha}\right)  >\tau_{v}\left(  1,y^{-\alpha}\right)  \}=0$,
then $D(\delta,t)=0$.

The following lemma is a key step in the proof of Theorem \ref{TH1-P-D1}.

\begin{lemma}
\label{case1}Under the conditions of Theorem \ref{TH1-P-D1}, we have
\[
\lim_{t\rightarrow \infty}\frac{\Pr \left(  X+Y>t,X\leq t,Y\leq \frac{t}%
{2}\right)  }{\left(  \overline{F}(t)\right)  ^{\kappa}\ell \left(
\overline{F}(t)\right)  }=\eta,
\]
where $\eta=\lim_{\delta \downarrow0}\eta(\delta)$.
\end{lemma}

\begin{proof}
Note that for any $0<\delta<1/2$,
\begin{align*}
&  \Pr \left(  X+Y>t,X\leq t,Y\leq \frac{t}{2}\right) \\
&  =\int_{0}^{1/2}\Pr \left(  \left.  t-ty<X<t\right \vert Y=ty\right)  dF(ty)\\
&  =\left(  \int_{\delta}^{1/2}+\int_{0}^{\delta}\right)  \Pr \left(  \left.
X>t-ty\right \vert Y=ty\right)  -\Pr \left(  \left.  X>t\right \vert Y=ty\right)
dF(ty)\\
&  =\left(  \int_{\delta}^{1/2}+\int_{0}^{\delta}\right)  \widehat{C}%
_{v}\left(  \overline{F}(t(1-y)),\overline{F}\left(  ty\right)  \right)
-\widehat{C}_{v}\left(  \overline{F}(t),\overline{F}\left(  ty\right)
\right)  dF(ty)\\
&  =I_{1}+I_{2},
\end{align*}
where the third step is due to the rewriting of $\Pr \left(  \left.
X>x\right \vert Y=y\right)  =\widehat{C}_{v}\left(  \overline{F}(x),\overline
{F}\left(  y\right)  \right)  $. We first consider term $I_{1}$. By Corollary
\ref{coro: partial derivative of copula regular varying-1}, for any
$\varepsilon>0$, there exists $t_{0}>0$ such that for $t>t_{0}$ we have
\[
\sup_{\delta<y\leq \frac{1}{2}}\left \vert \frac{\widehat{C}_{v}\left(
\overline{F}(t(1-y)),\overline{F}\left(  ty\right)  \right)  }{\left(
\overline{F}(t)\right)  ^{\kappa-1}\ell \left(  \overline{F}(t)\right)  }%
-\tau_{v}\left(  \left(  1-y\right)  ^{-\alpha},y^{-\alpha}\right)
\right \vert <\varepsilon,
\]
and%
\[
\sup_{\delta<y\leq \frac{1}{2}}\left \vert \frac{\widehat{C}_{v}\left(
\overline{F}(t),\overline{F}\left(  ty\right)  \right)  }{\left(  \overline
{F}(t)\right)  ^{\kappa-1}\ell \left(  \overline{F}(t)\right)  }-\tau
_{v}\left(  1,y^{-\alpha}\right)  \right \vert <\varepsilon.
\]
Then,
\begin{align*}
\frac{I_{1}}{\left(  \overline{F}(t)\right)  ^{\kappa-1}\ell \left(
\overline{F}(t)\right)  }  &  =\int_{\delta}^{1/2}\frac{\widehat{C}_{v}\left(
\overline{F}(t(1-y)),\overline{F}\left(  ty\right)  \right)  }{\left(
\overline{F}(t)\right)  ^{\kappa-1}\ell \left(  \overline{F}(t)\right)  }%
-\tau_{v}\left(  \left(  1-y\right)  ^{-\alpha},y^{-\alpha}\right)  dF(ty)\\
&  -\int_{\delta}^{1/2}\frac{\widehat{C}_{v}\left(  \overline{F}%
(t),\overline{F}\left(  ty\right)  \right)  }{\left(  \overline{F}(t)\right)
^{\kappa-1}\ell \left(  \overline{F}(t)\right)  }-\tau_{v}\left(  1,y^{-\alpha
}\right)  dF(ty)\\
&  +\int_{\delta}^{1/2}\tau_{v}\left(  \left(  1-y\right)  ^{-\alpha
},y^{-\alpha}\right)  -\tau_{v}\left(  1,y^{-\alpha}\right)  dF(ty)\\
&  \leq2\varepsilon \left(  \overline{F}(t\delta)-\overline{F}(t/2)\right)
+D(\delta,t).
\end{align*}
Since $D(\delta,t)\not =0$ for $0<\delta<1/2$, by the arbitrariness of
$\varepsilon$, we have that%
\[
\frac{I_{1}}{\left(  \overline{F}(t)\right)  ^{\kappa-1}\ell \left(
\overline{F}(t)\right)  }\lesssim D(\delta,t).
\]
The other side of the inequality can be obtained similarly. Thus,
\[
I_{1}\sim \eta(\delta)\left(  \overline{F}(t)\right)  ^{\kappa}\ell \left(
\overline{F}(t)\right)  .
\]
Since $\eta<\infty$, we have by Lemma \ref{D1D2}%
\[
\lim_{\delta \downarrow0}\lim_{t\rightarrow \infty}\frac{I_{1}}{\left(
\overline{F}(t)\right)  ^{\kappa}\ell \left(  \overline{F}(t)\right)  }=\eta.
\]

It remains to show that
\[
\lim_{\delta \downarrow0}\lim_{t\rightarrow \infty}\frac{I_{2}}{\left(
\overline{F}(t)\right)  ^{\kappa}\ell \left(  \overline{F}(t)\right)  }=0.
\]
First, by $\overline{F}\in \mathrm{RV}_{-\alpha}$ and Taylor's expansion, we
have for large $t$ and any $0\leq y\leq \delta t$,%
\[
\frac{\overline{F}(t-y)}{\overline{F}(t)}\sim \left(  1-\frac{y}{t}\right)
^{-\alpha}\sim1+\alpha \frac{y}{t}.
\]
Then by Assumption \ref{A3}, for any $\delta>0$ and all large $t$%
\begin{equation}
\frac{\widehat{C}_{v}\left(  \overline{F}(t-y),\overline{F}\left(  y\right)
\right)  }{\widehat{C}_{v}\left(  \overline{F}(t),\overline{F}\left(
y\right)  \right)  }-1=\frac{\widehat{C}_{v}\left(  \frac{\overline{F}%
(t-y)}{\overline{F}(t)}\overline{F}(t),\overline{F}\left(  y\right)  \right)
}{\widehat{C}_{v}\left(  \overline{F}(t),\overline{F}\left(  y\right)
\right)  }-1\lesssim c\frac{y}{t} \label{combine AS}%
\end{equation}
Then by applying (\ref{combine AS}) we have%
\begin{align*}
\lim_{\delta \downarrow0}\frac{I_{2}}{\left(  \overline{F}(t)\right)  ^{\kappa
}\ell \left(  \overline{F}(t)\right)  }  &  =\lim_{\delta \downarrow0}\int
_{0}^{\delta t}\left(  \frac{\widehat{C}_{v}\left(  \overline{F}%
(t-y),\overline{F}\left(  y\right)  \right)  }{\widehat{C}_{v}\left(
\overline{F}(t),\overline{F}\left(  y\right)  \right)  }-1\right)
\frac{\widehat{C}_{v}\left(  \overline{F}(t),\overline{F}\left(  y\right)
\right)  }{\left(  \overline{F}(t)\right)  ^{\kappa}\ell \left(  \overline
{F}(t)\right)  }dF(y)\\
&  \lesssim \lim_{\delta \downarrow0}c\int_{0}^{\delta}\frac{y}{t}\frac
{\widehat{C}_{v}\left(  \overline{F}(t),\overline{F}\left(  y\right)  \right)
}{\left(  \overline{F}(t)\right)  ^{\kappa}\ell \left(  \overline{F}(t)\right)
}dF(y)\\
&  \rightarrow0,
\end{align*}
where the last step is due to (\ref{A5}). This completes the proof.\bigskip
\end{proof}

The next lemma is needed in the proof of Theorem \ref{Thm2}.

\begin{lemma}
\label{term2}Under the conditions of Theorem \ref{Thm2}, we have as
$t\rightarrow \infty$,
\[
\Pr \left(  X+Y>t,X\leq t,Y\leq \frac{t}{2}\right)  =\left(  \overline
{F}(t)\right)  ^{\theta}h\left(  \overline{F}(t)\right)  \int_{0}^{1/2}\left(
(1-y)^{-\alpha \theta}-1\right)  \varphi(1,\overline{F}\left(  ty\right)
)dF(ty)(1+o(1)).
\]

\end{lemma}

\begin{proof}
Let $\varepsilon$ be arbitrarily small. Using the same rewriting in the proof
of Lemma \ref{case1}, we have
\[
\frac{\Pr \left(  X+Y>t,X\leq t,Y\leq \frac{t}{2}\right)  }{\left(  \overline
{F}(t)\right)  ^{\theta}h\left(  \overline{F}(t)\right)  }=\int_{0}%
^{t/2}\left(  \frac{\widehat{C}_{v}\left(  \overline{F}(t-y),\overline
{F}\left(  y\right)  \right)  }{\widehat{C}_{v}\left(  \overline
{F}(t),\overline{F}\left(  y\right)  \right)  }-1\right)  \frac{\widehat
{C}_{v}\left(  \overline{F}(t),\overline{F}\left(  y\right)  \right)
}{\left(  \overline{F}(t)\right)  ^{\theta}h\left(  \overline{F}(t)\right)
}dF(y).
\]
Then
\begin{align*}
&  \int_{0}^{t/2}\left(  \frac{\widehat{C}_{v}\left(  \overline{F}%
(t-y),\overline{F}\left(  y\right)  \right)  }{\widehat{C}_{v}\left(
\overline{F}(t),\overline{F}\left(  y\right)  \right)  }-1\right)
\frac{\widehat{C}_{v}\left(  \overline{F}(t),\overline{F}\left(  y\right)
\right)  }{\left(  \overline{F}(t)\right)  ^{\theta}h\left(  \overline
{F}(t)\right)  }dF(y)-\int_{0}^{t/2}\left(  \left(  1-\frac{y}{t}\right)
^{-\alpha \theta}-1\right)  \varphi \left(  1,\overline{F}\left(  y\right)
\right)  dF(y)\\
&  =\int_{0}^{\delta t}\left(  \frac{\widehat{C}_{v}\left(  \overline
{F}(t-y),\overline{F}\left(  y\right)  \right)  }{\widehat{C}_{v}\left(
\overline{F}(t),\overline{F}\left(  y\right)  \right)  }-\left(  1-\frac{y}%
{t}\right)  ^{-\alpha \theta}\right)  \frac{\widehat{C}_{v}\left(  \overline
{F}(t),\overline{F}\left(  y\right)  \right)  }{\left(  \overline
{F}(t)\right)  ^{\theta}h\left(  \overline{F}(t)\right)  }dF(y)\\
&  +\int_{\delta t}^{t/2}\left(  \frac{\widehat{C}_{v}\left(  \overline
{F}(t-y),\overline{F}\left(  y\right)  \right)  }{\widehat{C}_{v}\left(
\overline{F}(t),\overline{F}\left(  y\right)  \right)  }-1\right)
\frac{\widehat{C}_{v}\left(  \overline{F}(t),\overline{F}\left(  y\right)
\right)  }{\left(  \overline{F}(t)\right)  ^{\theta}h\left(  \overline
{F}(t)\right)  }dF(y)\\
&  +\int_{0}^{\delta t}\left(  \left(  1-\frac{y}{t}\right)  ^{-\alpha \theta
}-1\right)  \left(  \frac{\widehat{C}_{v}\left(  \overline{F}(t),\overline
{F}\left(  y\right)  \right)  }{\left(  \overline{F}(t)\right)  ^{\theta
}h\left(  \overline{F}(t)\right)  }-\varphi \left(  1,\overline{F}\left(
y\right)  \right)  \right)  dF(y)\\
&  -\int_{\delta t}^{t/2}\left(  \left(  1-\frac{y}{t}\right)  ^{-\alpha
\theta}-1\right)  \varphi \left(  1,\overline{F}\left(  y\right)  \right)
dF(y)\\
&  :=I_{1}+I_{2}+I_{3}+I_{4}.
\end{align*}
Let $\Phi(x)=\int_{0}^{x}\varphi(1,v)dv$ for $0\leq x\leq1$. Denote
$I=\int_{0}^{t/2}\left(  \left(  1-\frac{y}{t}\right)  ^{-\alpha \theta
}-1\right)  \varphi \left(  1,\overline{F}\left(  y\right)  \right)
dF(y)=-\int_{0}^{t/2}\left(  \left(  1-\frac{y}{t}\right)  ^{-\alpha \theta
}-1\right)  d\Phi(\overline{F}\left(  y\right)  )$. We will show that
$I_{k}=o\left(  I\right)  $ for $k=1,2,3,4$.

We first analyze $I_{4}$, which will be used in the other three terms. Since
$0\leq \beta \leq1$, by Assumption \ref{A4}\ and Karamata's Theorem, we have
$\Phi(\overline{F}\left(  \cdot \right)  )\in \mathrm{RV}_{-\alpha(1-\beta)}$.
Note that $\alpha(1-\beta)\geq1$. By Lemma 2.4 of Mao and Hu
\cite{Mao-Hu-2013},
\begin{equation}
I\sim-\frac{\alpha \theta}{t}\int_{0}^{t/2}yd\Phi(\overline{F}\left(  y\right)
), \label{prop I}%
\end{equation}
which implies that $I_{4}=o\left(  I\right)  $.

For $I_{1}$, applying (\ref{combine AS}) leads to%
\begin{align*}
I_{1}  &  \lesssim \int_{0}^{\delta t}\left(  \alpha(\theta+\varepsilon
)\frac{y}{t}-\left(  1-\frac{y}{t}\right)  ^{-\alpha \theta}\right)
\frac{\widehat{C}_{v}\left(  \overline{F}(t),\overline{F}\left(  y\right)
\right)  }{\left(  \overline{F}(t)\right)  ^{\theta}h\left(  \overline
{F}(t)\right)  }dF(y)\\
&  \lesssim \frac{\varepsilon \alpha}{t}\int_{0}^{\delta t}y\frac{\widehat
{C}_{v}\left(  \overline{F}(t),\overline{F}\left(  y\right)  \right)
}{\left(  \overline{F}(t)\right)  ^{\theta}h\left(  \overline{F}(t)\right)
}dF(y).
\end{align*}
Since there exists a function $g$ to bound $\frac{\widehat{C}_{v}\left(
\overline{F}(t),\overline{F}\left(  y\right)  \right)  }{\left(  \overline
{F}(t)\right)  ^{\theta}h\left(  \overline{F}(t)\right)  }$, by the dominated
convergence theorem
\[
\lim_{t\rightarrow \infty}\int_{0}^{\infty}\frac{\widehat{C}_{v}\left(
\overline{F}(t),\overline{F}\left(  y\right)  \right)  }{\left(  \overline
{F}(t)\right)  ^{\theta}h\left(  \overline{F}(t)\right)  }dF(y)=\int
_{0}^{\infty}\varphi \left(  1,\overline{F}\left(  y\right)  \right)  dF(y).
\]
Then by (\ref{prop I}),%
\[
\frac{I_{1}}{I}\lesssim \frac{\varepsilon \int_{0}^{\delta t}y\frac{\widehat
{C}_{v}\left(  \overline{F}(t),\overline{F}\left(  y\right)  \right)
}{\left(  \overline{F}(t)\right)  ^{\theta}h\left(  \overline{F}(t)\right)
}dF(y)}{-\theta \int_{0}^{t/2}yd\Phi(\overline{F}\left(  y\right)
)}\rightarrow \frac{\varepsilon}{^{\theta}}.
\]
The other side can be obtained similarly. By the arbitrariness of
$\varepsilon$, it leads to $I_{1}=o\left(  I\right)  $.

For $I_{2}$, by Corollary
\ref{coro: partial derivative of copula regular varying-1} and Lemma
\ref{D1D2},
\begin{align*}
I_{2}  &  =\int_{\delta}^{1/2}\left(  \frac{\widehat{C}_{v}\left(
\overline{F}(t(1-y)),\overline{F}\left(  ty\right)  \right)  }{\left(
\overline{F}(t)\right)  ^{\kappa}\ell \left(  \overline{F}(t)\right)  }%
-\frac{\widehat{C}_{v}\left(  \overline{F}(t),\overline{F}\left(  y\right)
\right)  }{\left(  \overline{F}(t)\right)  ^{\kappa}\ell \left(  \overline
{F}(t)\right)  }\right)  \frac{\left(  \overline{F}(t)\right)  ^{\kappa
-\theta}\ell \left(  \overline{F}(t)\right)  }{h\left(  \overline{F}(t)\right)
}dF(y)\\
&  \sim \eta(\delta)\frac{\left(  \overline{F}(t)\right)  ^{\kappa-\theta}%
\ell \left(  \overline{F}(t)\right)  }{h\left(  \overline{F}(t)\right)  }.
\end{align*}
Then by (\ref{prop I}) and (\ref{cond}), $I_{2}=o\left(  I\right)  $.

For $I_{3}$, similar with the analysis of $I_{1}$, we obtain $I_{3}=o\left(
I\right)  $. This ends the proof.
\end{proof}

\bigskip

Lastly, we prove the two main theorems.

\begin{proof}
[Proof of Theorem \ref{TH1-P-D1}]For $t>0$, consider the following
decomposition%
\begin{align}
\{X+Y>t\}  &  =\left \{  X+Y>t,0<Y\leq \frac{t}{2}\right \}  +\left \{
X+Y>t,0<X\leq \frac{t}{2}\right \}  +\left \{  X>\frac{t}{2},Y>\frac{t}%
{2}\right \} \nonumber \\
&  =:A_{1}+A_{2}+A_{3}. \label{decomp}%
\end{align}
By Lemma \ref{case1} and the symmetry of $\widehat{C}$, we have for $i=1,2$%
\begin{align*}
\Pr \left(  A_{i}\right)   &  =\overline{F}(t)-\widehat{C}\left(  \overline
{F}(t),\overline{F}\left(  \frac{t}{2}\right)  \right)  +\Pr \left(
X+Y>t,X\leq t,Y\leq \frac{t}{2}\right) \\
&  =\overline{F}(t)+\left(  \eta-\tau \left(  1,2^{\alpha}\right)  \right)
\left(  \overline{F}(t)\right)  ^{\kappa}\ell \left(  \overline{F}(t)\right)
(1+o(1)).
\end{align*}
For $A_{3}$,
\begin{equation}
\Pr \left(  A_{3}\right)  =\widehat{C}\left(  \overline{F}\left(  \frac{t}%
{2}\right)  ,\overline{F}\left(  \frac{t}{2}\right)  \right)  \sim \tau \left(
2^{\alpha},2^{\alpha}\right)  \left(  \overline{F}(t)\right)  ^{\kappa}%
\ell \left(  \overline{F}(t)\right)  . \label{third}%
\end{equation}
Combining the above analysis for $A_{1}$, $A_{2}$ and $A_{3}$ gives the
desired result. \bigskip
\end{proof}

\begin{proof}
[Proof of Theorem \ref{Thm2}]We apply the decomposition in (\ref{decomp}). By
Lemma \ref{term2} and the symmetry of $\widehat{C}$, we have for $i=1,2$,%
\begin{align*}
\Pr \left(  A_{i}\right)   &  =\overline{F}(t)-\tau \left(  1,2^{\alpha}\right)
\left(  \overline{F}(t)\right)  ^{\kappa}\ell \left(  \overline{F}(t)\right)
(1+o(1))\\
&  +\int_{0}^{1/2}\left(  (1-y)^{-\alpha \theta}-1\right)  \varphi
(1,\overline{F}\left(  ty\right)  )dF(ty)\left(  \overline{F}(t)\right)
^{\theta}h\left(  \overline{F}(t)\right)  (1+o(1)).
\end{align*}
The analysis of $A_{3}$ is the same as given in (\ref{third}). Combining the
above analysis for $A_{1}$, $A_{2}$ and $A_{3}$ gives the desired
result.\bigskip
\end{proof}

\textbf{Availability of Data and Materials} Data sharing not applicable to
this article as no datasets were generated or analyzed during the current study.

\end{document}